\documentclass[a4paper,UKenglish,cleveref, autoref, thm-restate]{sn-jnl}
\usepackage[utf8]{inputenc} % allow utf-8 input
\usepackage{url,amsmath}            % simple URL typesetting
\usepackage{mathtools}
\usepackage{booktabs}       % professional-quality tables
\usepackage{amsfonts,amsthm}       % blackboard math symbols
 \usepackage{graphicx,longtable}
\usepackage{algorithm2e}

\newcommand{\thmtexta}{If $f$  and $g$ are $(3\eps_1, \eps_2)$-half-stable and $f \not \sim_{(3\eps_1, d.\eps_2)} g$ for $d\geq 3$ then there exists an  $(\eps_1, (d-2). \eps_2)$-rectangle which separates $f$ and $g$.
}
 
\newcommand{\thmtext}{Let   
$\eps_1,\eps_2, \delta$, 
a frequency function $f$ and 
a stream $s$ with insertions only be given. 
If the frequency functions $f$ and $g$ are $ (3\eps_1, \eps_2)$\hbox{-half-stable} and $(\gamma_1,\gamma_2)$-decreasing 
then Algorithm $A_1(s,\eps_1,\eps_2,f)$
 is a streaming Tester for $f$ that uses space $O(\log^2 n\cdot \log\log n)$. It accepts with high probability if $f=g$ and rejects if $f \not\sim_{ (3.\eps_1,10\eps_2)} g$ with high probability.}
 
\newcommand{\thmtextb}{Let   
$\eps_1,\eps_2, \delta$, 
and 
a stream $s$ of elements of dimension $d$ with insertions only be given. 
 If the marginal distributions $g_1$ on the first $d/2$ elements and   
 $g_2$ on the last $d/2$ elements satisfy the same hypothesis as in theorem \ref {main3} ($g_1$ is   $(3\eps_1/2, \eps_2/4)$ half-stable  and  $(\gamma_1/(1+\eps_1),(1+\eps_2)^2\gamma_2 )$-decreasing and $g_2$ is  $ (3\eps_1,\eps_2)$\hbox{-half-stable} and $(\gamma_1,\gamma_2)$-decreasing )
then Algorithm $A_3(s,\eps_1,\eps_2, \delta, \pi_1, \pi_2)$
 is a streaming Tester that uses space 
 $O(\log^2 n\cdot \log\log n)$. It accepts with high probability if $g_1=g_2$ and rejects if $g_1 \not\sim_{ (4.\eps_1,12\eps_2)} g_2$ with high probability.}

\newcommand{\thmtextc}{Let   
$\eps_1,\eps_2, \delta$, 
and two streams $s_1$  and $s_2$ of elements of
domains $D_1$ and $D_2$ of the same size $n$, and let $g_1$ and $g_2$ be their frequency distributions. 
 If   $g_1$ is   $(3\eps_1/2, \eps_2/4)$ half-stable  and  $(\gamma_1/(1+\eps_1),(1+\eps_2)^2\gamma_2 )$-decreasing and if $g_2$ is  $ (3\eps_1,\eps_2)$\hbox{-half-stable} and $(\gamma_1,\gamma_2)$-decreasing 
then Algorithm $A_2(\eps_1,\eps_2,\delta)$
 is a streaming Tester that uses space $O(\log^2 n\cdot \log\log n)$. It accepts with high probability if $g_1=g_2$ and rejects if $g_1 \not\sim_{ (4.\eps_1,12\eps_2)} g_2$ with high probability.}
 
\newtheorem*{T1}{Theorem~\ref{main}}

\newtheorem*{T3}{Theorem~\ref{main3}}
\newtheorem{notation}{Notation} 
\newcommand{\dist}{\mathsf{dist}}
 
\newcommand{\E}{I\!\!E}

\newcommand{\ccount}{\hbox{count}}
\newcommand{\occ}{\hbox{occ}}

\newcommand{\rank}{\hbox{rank}}
   \newcommand{\VDR}{\mathsf{VDR}}
     \newcommand{\EMD}{\mathsf{EMD}}
\newsavebox{\fmbox}
\newcommand{\eps}{\varepsilon}
\newtheorem{theorem-non}{Theorem}
\newtheorem{lemma}{Lemma}
\newtheorem{definition}{Definition}
\newtheorem{corollary}{Corollary}

\title{Testing frequency distributions in a stream}
\author*[1]{\fnm{Claire} \sur{Mathieu}}\email{claire@irif.fr}
\author*[1,2]{\fnm{Michel} \sur{de Rougemont}}\email{mdr@irif.fr}
\affil*[1]{\orgdiv{IRIF}, \orgname{CNRS}, \orgaddress{\street{ }\city{Paris},  \postcode{ } \country{France} \state{ } }}

\affil*[2]{\orgdiv{IRIF }, \orgname{University Paris II}, \orgaddress{\street{ } \city{Paris}, \postcode{ }  \country{France} \state{ } }}

\begin{document}

\maketitle

\begin{abstract}

How does one verify that the frequency function $g$ given by a stream of $N$ data items taken from a universe of $n$ distinct items closely matches a given frequency function $f$? We introduce the \emph{relative Fr{\'e}chet distance} to compare two frequency functions  in a homogeneous manner. 
%We consider two streaming models: insertions only and sliding windows. 
We then present a Tester  for a certain class of functions, which decides if $f=g$ or if $f$ is  far from $g$ with high probability. The Tester only uses polylogarithmic space. We prove that the assumptions defining the class are necessary by giving linear space lower bounds when some assumption fails.  %If $f$ is uniform we show a space $\Omega(n)$ lower bound. If $f$ decreases fast enough, the Tester then only use space $O(\log^2 n\cdot \log\log n)$. 
We then generalize the Tester to compare the frequency functions of two distinct streams. If the elements of the stream are tuples of dimension $d$, we can also  compare the frequencies of two marginal distributions. The techniques  rely on the Spacesaving algorithm \cite{MAE2005,Z22} and on sampling the stream. 
\end{abstract}
\maketitle
\newpage
\tableofcontents
\newpage
\section{Introduction}

The frequency function $g$ of a stream of data items is such that $g(i)$ is the number of occurrences of the $i$th most frequent item in the stream.  Here, we consider a stream of length $N$, consisting of elements from a domain $U$ of size $n$ and we want to verify whether its frequency function $g$ matches a fixed distribution $f$. We may also look at two different streams and ask whether their frequencies $g_1$ and $g_2$ are close to each other.  In practice, of particular interest are settings with single-pass streams and very small memory~\cite{M02}. 
 
 %SpaceSaving
 The well-known SpaceSaving algorithm \cite{MAE2005} finds the most frequent items (heavy hitters) using very small memory. 
 Here, we design an extension of SpaceSaving to identify, not just the frequencies of the heavy hitters, but also the frequencies of the lower frequency items in the tail of the distribution, and more generally the entire profile of the frequency function. 
 Our algorithm uses SpaceSaving as a black box and simply  leverages it with probabilistic subsampling of the items in the stream.
 What kind of properties can we hope to verify if we only allow poly-logarithmic space? 
 
 Our first result is negative. 
 We prove a linear space lower bound $\Omega (n)$ on the space of the Tester, Theorem \ref{lb},  to verify whether the frequency function of the stream is "close"  to the uniform function. We therefore need some additional assumption. In this paper, we restrict attention to frequency functions that decrease "fast enough", see definition~\ref{definition:decreasing}. 
 
The rate of decrease does not suffice. Our second result is also negative.  We also prove a linear space lower bound $\Omega (n)$ on the space of the Tester, Theorem \ref{lb2},  to verify whether the frequency function of the stream is "close"  to a fast decreasing function that has some particular "double  discontinuity"   feature. We therefore need some additional assumption. In this paper, we further restrict attention to frequency functions that have no such double discontinuities, see definition~\ref{hs} (half-stability). That insures that the frequency function  may have discontinuities, i.e. large drops, but no double discontinuities. 
These two assumptions include real-life functions such as power-law distributions. 
%\marginpar{\small all power law or some power law with some value of $\alpha$? Zipf?}

Given a stream and a fixed frequency function $f$ which is half-stable and decreasing fast enough,   the objective is to decide whether the frequency defined by the stream is close to $f$, following the Property Testing framework. What does it mean for two frequency functions to be "close" ? We use a new measure of distance, the \emph{relative Fr{\'e}chet distance}. 
It is an adaptation of the classical (discrete) Fr{\'e}chet distance, which is an absolute distance, studied in particular in computational geometry, including in the streaming context~\cite{D19,F20}. The { relative Fr{\'e}chet distance} must preserve  the distance on the $x$-axis within $(1+\eps_1)$ and  the distance on the $y$-axis within $(1+\eps_2 )$.

  We combined two well known techniques: the Spacesaving algorithm \cite{MAE2005,Z22} which deterministically selects the most frequent items approximately and the Minhash technique which approximates the low frequencies probabilistically. 
  
  All of the following results assume that the functions under study are half-stable (since otherwise there is a linear lower bound, Theorem \ref{lb2}). A preliminary result  shows that
  any two frequency functions either are close to each other (in terms of relative Fr\'echet distance) or have a separating rectangle between them, Theorem \ref{main0}, 
  
\paragraph{Main results.} 

\begin{itemize}
%\item  A finer analysis of the Spacesaving algorithm, lemma \ref{ss}, in the appendix \ref{ssa},
% \item Any two frequency functions either are close to each other (in terms of relative Fr\'echet distance) or have a separating rectangle between them, Theorem \ref{main0}, 
 
 \item A  streaming Tester $A_1$ to compare the frequency of a stream to a fixed reference frequency function, both  $(\gamma_1,\gamma_2)$-decreasing, for  the relative Fr{\'e}chet distance. The Tester uses $O(\log^2 n\cdot \log\log n)$ space, Theorem \ref{main}. The Tester uses the separating rectangle as a witness for negative instances.
 
 \item A  streaming Tester $A_2$ to compare the frequencies of two streams $g_1$ and $g_2$, both  $(\gamma_1,\gamma_2)$-decreasing,  for  the relative Fr{\'e}chet distance. The Tester also uses  $O(\log^2 n\cdot \log\log n)$ space, Theorem \ref{main3}.

 \end{itemize}

There are two different representations of frequencies. Either the function $F$ is from  the universe $U$ to ${\bf N}$ and $F(e)$ is the frequency of  item $e$; or the function $f$ is from  $\{1,2,..|U|\}$ to   ${\bf N}$ and $f(i)$ is the frequency of the $i$-th most frequent item; we take the latter viewpoint. A \emph{frequency function} $g$ is a non-negative integer-valued function over a set of elements such that $f(i)$ is the number of occurences of the $i$th most frequent element. For example, the two streams $aaabba$ and $bbbaab$ have the same frequency functions ($f(1)=4$, $f(2)=2$) even though element $b$ has 2 occurrences in the first stream and 4 occurrences in the second stream. %The streams $aabba$ and $xyxxy$ have the same frequency functions on different domains  of size $n=2$ but the frequencies of all items $\{a,b,x,y\}$ are different in the two streams.

The second section 
presents our main definitions and we present the Testers in the third section. We give the space lower bounds $\Omega(n)$ in the fourth section. In the fifth section, we review the bounds
 associated with the Spacesaving algorithms  and a fine analysis, lemma \ref{ss}.   
In the sixth section
we analyze the Testers and prove Theorem \ref{main}  and Theorem \ref{main3}.
%Properties of  half-stable and $(\gamma_1,\gamma_2)$-decreasing functions are in the appendix \ref{scgd}. 
%The  proof of Theorem \ref{fml} is in the appendix \ref{b1} and the proof of Theorem \ref{main2}, the generalization to compare two marginal distributions on elements of dimension $d$, is in the appendix \ref{comp}.

\paragraph{Motivations and comparison with other approaches.}
Problems that are hard in the worst-case may be much simpler for inputs which follow specific distributions, for example power law distributions. It is therefore important to verify if some given data follow certain distributions, % and estimate the parameters of the distributions, 
when the data arrive in a stream. If each item is a tuple of dimension $d$, there are many marginal distributions which we might want to compare, without a given reference distribution $f$.

The area of {\em Distribution testing} \cite{C20} studies this type of problems in general and assume a model where we sample the distribution, whereas we don't use this feature. The role of the uniform distribution is then very different. The distance between distributions is also important in this context.

The $L_1$ distance, also called the {\em variational distance} between two distributions $ \delta_1, \delta_2  $ on a domain with $n$ elements is defined as: $$\dist_1( \delta_1, \delta_2   )=\frac{1}{2} \cdot\sum_i  \mid  \delta_1(i)- \delta_2(i)   \mid . $$ 
If we relabel the domain with  a permutation $\pi$ we may have a smaller variation $\sum_i  \mid  \delta_1(i)- \delta_2(\pi(i))   \mid $ and 
\cite{GR20} introduces  the {\em Variation distance up to relabeling} $\VDR(\delta_1, \delta_2  )$\footnote{For a distribution $\delta$, let the histogram  $h_{\delta}$ of $\delta$ be the function  $[0,1] \rightarrow N$ such that $h_{\delta}(x)=|\{i:~\delta(i)=x\}|$, the number ef elements with probability $x$. Then \cite{GR20} shows the connection betwen $\VDR$ and the  Earth-Moving Distance $\EMD$ of the histograms: 
$\VDR(\delta_1, \delta_2  )=\EMD(h_{\delta_1},h_{\delta_2})/2$. } as the minimum over $\pi$ of  $\frac{1}{2} \cdot\sum_i  \mid  \delta_1(i)- \delta_2(\pi(i))   \mid  $. It is also the $L_1$ distance between the frequencies of the distributions, i.e. ordered by decreasing values and the distribution of the frequencies is invariant by relabeling.

The classical Spacesaving algorithm \cite{MAE2005} gives estimates on the most frequent items, but no information on the tail of the distribution. The proposed Tester compares the most frequent items on some randomized substreams using the Spacesaving algorithm as a black box and indirectly gives some information on the tail of the frequency distribution. It assumes additional hypothesis on the function $g$, to be half-stable and $(\gamma_1,\gamma_2)$-decreasing, in order to obtain relative errors on the frequencies, as in \cite{C23} to approximate the rank of an item. Classical 
sketching approaches such as  {\em Count-Min Sketch} \cite{C05}
only guarantee an additive error. As we want to approximate small values in the tail of the distribution, we need relative errors.
We work in the streaming insertion model.  

In \cite{C19}, the verification of properties of a stream is studied  with streaming interactive proofs. In \cite{C18}, the verification is done efficiently thanks to prior work done by annotating the stream in advance in preparation for the task. In our setting, we use the Property Testing framework without any annotations or other additional prior information, as in  \cite{F02}. We propose this setting for the verification of the frequency distributions, given a stream of data.

A standard problem in statistics is to check if some observed data, i.e. in the insertion only model, approximately fit 
some statistics $F$ where $F(e_i) $ is the frequency of the element $e_i$. Let $G$ be the frequency of the elements of the observed data. The standard $\chi^2$ test compares two statistics\footnote{Let
$\chi^2(F,G)= \sum_{i=1}^{n} (F(e_i)-G(e_i))^2/F(e_i)$.

 If  $\chi^2(F,G) \leq a$, we know that $G$ follows $F$ with  confidence $1-\alpha$, for example $a=11,07$ and $1-\alpha=95\%$. }. In  this setting, \cite{FL20} gives an algorithm  which uses space $O(\log N. \sqrt{N})$ to decide if $F$ and $G$ are close or far for the $\chi^2$ test. In fact, the AMS-sketch \cite{A99} can be adapted and requires only $O(\log n)$ space.  In this case, the problem is easy because the statistics $F$ takes an element as argument.
 
% In this paper, we first study the case when the frequency function $g$ is given by a stream of $N$ data items and we want  to test if$g$ approximately follows the frequency function $f$ over the domain $\{1,2,...n\}$, in polylogarithmic space and without necessarily knowing the exact value of $n$.  For example $f$ might be a Zipf distribution. 
 %If we observe  sliding windows of the stream,  the frequency $g$ may be stable in each window, although the most frequent items change over time.Thus we are interested in making restrictive but reasonable assumptions that will imply that we can test in polylogarithmic space. 

%This paper belongs to the field of distribution testing, see for example \cite{C20} for a general exposition of the domain.  

\section{Main definitions}

The main definitions are the {\em relative Fr\'echet distance} definition \ref{definition:Frechet}, the 
$(\gamma_1,\gamma_2)$-decreasing functions
definition \ref{definition:decreasing} and the $(\eps_1,\eps_2)$-half-stable functions, definition \ref{hs}. The third one is the most novel concept.

 \subsection{Relative Fr\'echet distance}

We first define a relative distance between points. 

\begin{definition}\label{defi:epsclose}
Let $0< \eps_1,\eps_2<1$.  We say that two non-negative numbers $a,b$ are \emph{$\eps$-close}, and denote it by $a \simeq_{\eps} b$, if $|  a-b   | \leq \eps \cdot \min\{a,b\}$.  We say that two points $p=(x,y)$ and $p'=(x',y')$ are $(\eps_1,\eps_2)$-close, and denote it by $p\simeq_{(\eps_1,\eps_2)}p'$,  if $x\simeq_{\eps_1}x'$ and $y\simeq_{\eps_2}y'$.  \end{definition}

\begin{lemma}\label{lemma:losing-epsilon}
If $|a-b|\leq \epsilon \cdot a$ then $a\simeq_{\eps/(1-\eps)} b$.
If $a\not\simeq_{\eps}b$ then $|a-b|>\frac{\eps}{1+\eps}\cdot \max(a,b).$
\end{lemma} 
\begin{proof}
If $a\leq b$ then $|a-b|\leq \epsilon \cdot a\leq \epsilon \cdot \min(a,b)$ and the first part of the Lemma holds.
If $b<a$ then the hypothesis says that $|a-b|=a-b<\epsilon \cdot a$, so $a(1-\epsilon)<b$, implying $\epsilon \cdot a < (\epsilon/1-\epsilon)\cdot b$ and the first part of the Lemma also holds.

For the second part, the hypothesis states that $|a-b|>\eps \cdot \min(a,b)$. 
Assume, by symmetry, that $a\leq b$. 
Then $b-a>\eps \cdot a$, 
so $b>a \cdot (1+\eps)$. 
Rewriting, $a<\frac{1}{1+\eps} \cdot b$, and so, 
$b-a>b-\frac{1}{1+\eps} \cdot b=\frac{\eps}{1+\eps} \cdot b$, thus 
$|a-b|>\frac{\eps}{1+\eps} \cdot \max(a,b).$
\end{proof}

\begin{definition}(Relative Fr{\'e}chet distance)\label{definition:Frechet}
Let $f$ and $g$ be two functions with domain $\{ 1,\cdots ,n\}$.  We say that $f$ and $g$ are $(\eps_1,\eps_2)$-close, denoted $f \sim_{(\eps_1,\eps_2)} g$,  if for every point $p=(x,f(x))$ of $f$ there exists an $(\eps_1,\eps_2)$-close point $p'=(x',g(x'))$ of $g$ and symmetrically for every point $p=(x,g(x))$ of $g$ there exists an $(\eps_1,\eps_2)$-close point $p'=(x',f(x'))$ of $f$.
In other words:
\begin{itemize}
\item $\forall p=(x,f(x)),~~ \exists p'=(x',g(x'))~:~ p\sim_{(\eps_1,\eps_2)}p'$
\item $\forall p=(x,g(x)),~~ \exists p'=(x',f(x'))~:~ p\sim_{(\eps_1,\eps_2)}p'$
\end{itemize}

\end{definition}

In appendix \ref{a1} we show that this definition is equivalent to  a relative Fr{\'e}chet distance based on a coupling between $f$ and $g$.
%Note that unlike the absolute Fr\'echet distance, the relative Fr\'echet distance is invariant by scaling.
The relation $f \sim_{(\eps_1,\eps_2)} g$ is reflexive and symmetric. 
 
\subsection{$(\gamma_1,\gamma_2)$-decreasing frequency functions}

We want the frequency function to decrease fast enough in a precise way.\\

\begin{definition}$((\gamma_1,\gamma_2)$-decreasing) Let $\gamma_2>\gamma_1 >1$. 
A non-increasing function $f$ with domain $\{ 1,\cdots ,n\}$
 is $(\gamma_1,\gamma_2 )$-decreasing if for all $t$ such that $1 \leq \gamma_1. t \leq n$:
$$f(\lceil \gamma_1.t \rceil )\leq f(t)/\gamma_2$$
\label{definition:decreasing}
\end{definition}

%\begin{definition}($(\gamma_1,\gamma_2)$-decreasing) Let $\gamma >1$. 
%A non-increasing function $f$ with domain $\{ 1,\cdots ,n\}$
% is $(\gamma_1,\gamma_2)$-decreasing if for all $t$ such that $1 \leq \gamma. t \leq n$:
%$f(\lceil \gamma.t \rceil )\leq f(t)/2$.
%\label{definition:decreasing}
%\end{definition}
 With $(\gamma_1,\gamma_2)$-decreasing  frequency functions, we can  bound the tail of the frequency distribution, Lemma~\ref{tail}, and bound the relative errors incurred by the counters in the SpaceSaving algorithm, Lemma~\ref{decreas}. 
 It is a property held by common distributions:  Zipf distributions are $(\gamma_1,\gamma_2)$-decreasing.

\subsection{Half-Stability}

The frequency function may have discontinuities, but we assume that it is ``not too discontinuous" in a way that we now define. \\

\begin{definition}(Rectangle)\label{rect}
Let $0<\eps_1,\eps_2<1$. An {\em $(\eps_1,\eps_2)$-rectangle} is a set $R \subseteq [1,n]\times [0,\infty]$ with  bottom left corner $(x,y)$ and top right corner $(x(1+\eps_1),y(1+\eps_2))$. \\

\end{definition}

\begin{definition}(Half-stability)\label{hs}
A frequency function $f$  
with domain $\{ 1,\cdots ,n\}$
is {\em $(\eps_1,\eps_2)$-half-stable}  if for every $t$, $1\leq t\leq n$, there exists an  $(\eps_1,\eps_2)$-rectangle $R$ containing $(t,f(t))$ and all the points of $f$ within the horizontal span of $R$. In other words:
$$\forall t\in [1,n]~~\exists (x,y) ~:~ t\in [x,x(1+\eps_1)] \hbox{ and } \forall z\in [x,x(1+\eps_1)], f(z)\in [y,y(1+\eps_2)].$$

\end{definition}

We can interpret $R$ as a window  $[x, x.(1+\eps_1)] $  which includes $t$ such that the values of the frequency function $f$ over the window are  in the interval 
$[y,y(1+\eps_2)]$.  In this case $R=
[x, x.(1+\eps_1)] \times [y, y(1+\eps_2) ]$.  

The function $f$ may have discontinuities but no double discontinuities: $f$ could have a discontinuity to the left of $t$, for example between $t-1$ and $t$, but then the values of $f$ to the right of $t$ are close to $f(t)$; or there could be a discontinuity to the right of $t$, for example between $t$ and $t+1$, but then the values of $f$ to the left of $t$ are close to $f(t)$.
If  $f$ is half-stable then every $t$ belongs to the horizontal span of a certain rectangle $R$, and all elements whose frequency is  approximately close can therefore be substituted for one another, hence sampling is a useful algorithmic approach.
%the $z_i$th most frequent element of $s$ might be absent from $s_i$, but one of its possible substitutes will (probably) be in $s_i$. This deals with the first source of inexactitude.

Is half-stability a reasonable assumption? Zipf distributions assume the relative frequency on a stream of length $m$, $f(i)/m=\frac{c}{i^{\alpha}}$ for $\alpha >1$, and power laws consider probabilities with a similar assumption. We ignore rounding problems as each $f(i)$ is an integer value.
Power laws and Zipf distributions are {\em $(\eps,\eps')$-half-stable}. (On the other hand a geometric distribution such as $f(i)=c/2^i$ is not {\em half-stable}, as it has large consecutive discontinuities.)

\begin{lemma} \label{zipf}
If $f$  is the relative frequency is  a Zipf distribution of parameter $\alpha$, then 
$f$ is  $(\eps/\alpha,\eps)$-half-stable.
\end{lemma}

\begin{proof}
  Let us find $j>i$ such that $f(j)\simeq f(i)/1+\eps$.
 We have:
 $$\frac{f(j)}{m}=\frac{c}{j^{\alpha}}\simeq \frac{c}{i^{\alpha}.(1+\eps)}$$
 Then $j \simeq i.(1+\eps)^{1/\alpha}\simeq i.(1+\eps/\alpha)$.
 \end{proof}
 
 In the appendix \ref{hsf} we give a characterization of the half-stable frequency functions.

\begin{definition}
We say that  a rectangle $ (x, x.(1+\eps_1) ), (y, y. (1+\eps_2))$ {\em separates}  two  functions $f$ and $g$ with domain $\{1,\ldots ,n\}$  if 
$$\forall j\in (x,x(1+\eps_1))~~ g(j) \leq y ~~\leq~~  y(1+\eps_2) \leq f(j) $$
 or conversely (exchanging $f$ and $g$). 
\end{definition}

In other words, $f$ is below the rectangle $R$ and $g $ is above $R$. No points $(t,f(t))$ of $f$ or  $(t,g(t))$ of $g$ is in $R$.
Notice that the point $(t,f(t))$ is the left of the rectangle for $t=1$ and at the right of the rectangle for $t=n$. 
The important structural property of half-stable functions is the following result which provides a witness that two half-stable functions are far from each other.\\

\begin{theorem-non}
  \label{main0}
  \thmtexta\\
\end{theorem-non}

\begin{proof}  
%(Proof of Theorem~\ref{})

%By contraposition of  Lemma \ref{sl}, and up to symmetry between 
By definition \ref{definition:Frechet} on the distance between $f$ and $g$, there exists a point $u=(i,f(i))$ of $f$ such that no point of $g$ is $(3\eps_1, d.\eps_2)$-close to it. 
All the points of $g$ are outside the rectangle 
$$R_d = 
  [i/(1+3\eps_1),i.(1+3\eps_1) ] * [f(i)/(1+d.\eps_2) ,f(i).(1+d.\eps_2) ]$$ which includes the rectangle $R_s$ defined below.
  
Since $f$ is $(3\eps_1, \eps_2)$-half-stable, there exist $x,y$ such that:

$$ x \leq i \leq x.(1+3\eps_1)$$
$$y\leq f(i) \leq y.(1+\eps_2)$$
and the points of $f$  
are inside the rectangle $R_s = 
  [x,x.(1+3\eps_1) ] * [y,y.(1+\eps_2) ]$.
  
  Notice that $R_s \subseteq R_d$ since  $d>1$. The points $(j,g(j)) $ are outside $R_d$ and there are two cases: either the curve $g$ does or does not cross the  rectangle $R_s$.
  It crosses  $R_s$ if there is a point $t$ such that:
  $$  x \leq t \leq x.(1+3\eps_1)  $$
  $$ g(t)\leq f(i)/(1+3\eps_2)$$
  $$f(i).(1+3\eps_2) \leq g(t-1)$$

In the first case, if $g$ does not cross  $R_s$, it is either above or below. Assume it is below, then the rectangle below $R_s$ in  $R_d$, i.e. 
  $$R_0= [x,x.(1+3\eps_1) ] * [f(i)/(1+d.\eps_2),y ]$$
  is an  $(3\eps_1,(d-2). \eps_2)$-rectangle separating $f$ and $g$. Its relative width is 
  $(1+3\eps_1)$ and its relative height is $(1+(d-2).\eps_2)$ as $y\geq f(i)/(1+\eps_2)$ and 
  $$y/ (f(i)/(1+d.\eps_2))\leq (1+d.\eps_2)/(1+\eps_2) \leq (1+(d-2). \eps_2)$$
  
  In the second case, if $g$ crosses  $R_s$, then consider the two rectangles $R_1$ above $R_s$ and $R_2$ below $R_s$ within the span of $R_s$ on each side of $t$, as shown in figure \ref{fig:sr}:

$$R_1= [x,t ] * [y.(1+\eps_2),f(i).(1+d.\eps_2)  ]$$
$$R_2= [t,x.(1+3\eps_1)] * [y,f(i).(1+d.\eps_2)  ]$$
Their relative height is larger than $(1+(d-2).\eps_2)$ because 
$$f(i).(1+d\eps_2)/ y.(1+\eps_2)> (1+d.\eps_2)/ (1+\eps_2) > (1+(d-2).\eps_2)$$
At least one of them has a relative width larger than $(1+\eps_1)$ because the product of their relative width is greater then $ (t/x )* (x.(1+3\eps_1)/t )=(1+3\eps_1)$. At least one of the rectangle has a width greater than $\sqrt{1+3\eps_1}>1+\eps_1$. Hence there is a  $(\eps_1,(d-2). \eps_2)$-rectangle separating $f$ and $g$.
\end{proof} 

 \begin{figure}
\includegraphics[width=120mm]{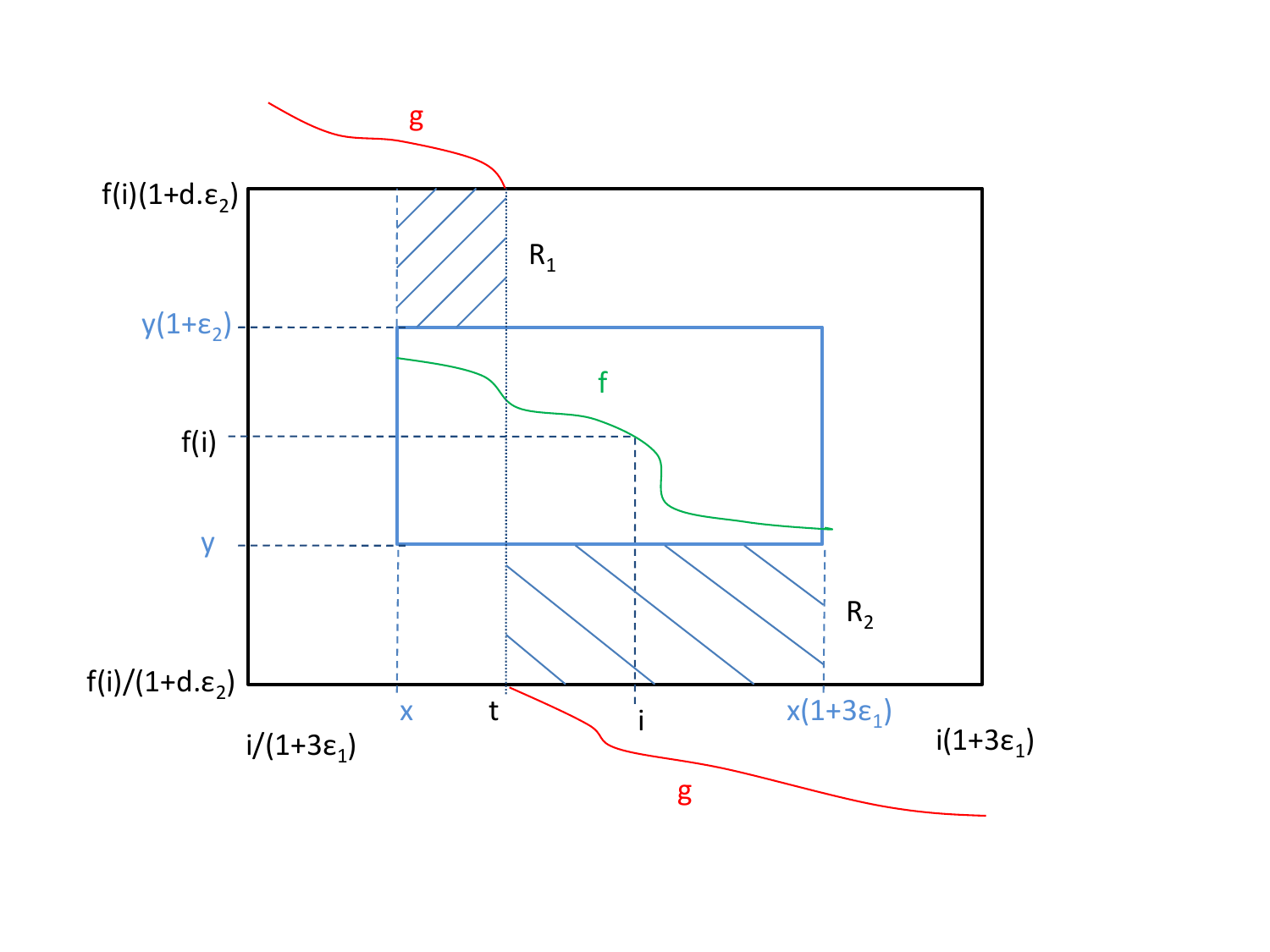}
\caption{Point $(i,f(i))$ is far from $g$, so rectangle $R_d$ (in black), that is centered on $(i,f(i))$, contains no point of $g$. Rectangle $R_d$   includes the smaller rectangle $R_s$ (in blue) that includes $f$ within its span. Examples of functions $f$ (within $R_s$) and $g$ (outside $R_d$) sre drawn.There are two possible separating rectangles, $R_1$  or $R_2$ (shaded). See the proof of Theorem \ref{main0}.}
\label{fig:sr}
\end{figure}
\section{Testers}

We first present a Tester \ref{alg:Tester} when a reference frequency $f$ is given and we want to compare the frequency function $g$ of a stream $s$ with the frequency $f$. We then consider a Tester 
\ref{alg:Tester2} when we are given two streams $s_1$ and $s_2$ and we compare the two frequency functions $g_1$ and $g_2$. We finally consider a Tester \ref{alg:Tester3} where elements are tuples of dimension $d$ and we compare the frequency functions of two marginal distributions.

The SpaceSaving algorithm was introduced in~\cite{MAE2005}  to compute estimates of 
the frequencies of the $k$ most frequent elements in a stream of elements from a universe of size $n$, using a table $T$ with $K\leq n$ entries. Each table entry consists of an element and a counter (plus some auxiliary information), which is a rough estimate of the frequency of the element in the stream. 

The SpaceSaving algorithm is straightforward: the table is kept sorted by counters: $c_1\geq c_2\geq \cdots c_K$. If the next element $e$ of the stream is in $T$, then the algorithm increments the corresponding counter; otherwise, it substitutes $e$ for the element whose counter is minimum ($c_K$ in position $K$), and increments the corresponding counter.  The value of $c_i$ is an estimate of $f(i)$. See \ref{rb} for details.

The following additive error result was proved in the original paper. (Note that $f(i)$ is the $i$th largest frequency whereas $c_i$ is the $i$th largest counter, so they count occurences of different elements in general). 

\begin{lemma}\cite{MAE2005} \label{sb}
Let $K$ denote the size of the table, $N$ denote the length of the stream, and $c_K$ the variable defined in the Space Saving algorithm. Then for every $i\leq K$ we have $|f(i)-c_i|\leq c_K$ and for every $i>K$ we have $f(i)\leq c_K$; moreover, $c_K \leq N/K$.
\end{lemma}

%**************************

Throughout the paper, we use $\ccount(e)$ to denote the value of the counter of element $e$. 
 Another elementary property of the SpaceSaving algorithm, the SpaceSaving Lemma, gives precise bounds on the error. The proof can be found in  section \ref{rb}, with a relative error bound Lemma \ref{decreas}.\\
 
\begin{lemma}[SpaceSaving Lemma]\label{lemma:diffcounts}
Let $K$ denote the size of the table and $c_K$ be as defined in the SpaceSaving algorithm.
Fix $r\leq K$. Let $\tilde{e}$ denote the element with the $r$th largest count in the table, and $e'$ denote the element with the $r$th largest frequency in $s$. Then
$$|\ccount(e')-\ccount(\tilde{e})|\leq c_K.$$
\end{lemma}

The SpaceSaving algorithm is  used when focusing on the frequencies of the top few elements of the universe. Here, we would like to leverage the power of the SpaceSaving algorithm to test whether the \emph{entire} distribution of frequencies of the stream approximates a given frequency distribution, with small \emph{relative} error.   For example, this can be used to check whether a stream of graph edges defines a graph whose degree sequence is close to a predicted degree sequence. 

%\subsection{Algorithm \ref{alg:Tester} }
With that, we can describe our streaming algorithm to test whether the frequency distribution $g$ defined by the elements of a stream is close to a specified frequency distribution $f$.
Let $z_{i}=\lceil(1+\eps_1^2)^i  \rceil$ for $i\geq 1$. 
%The frequency function $g$ is $(3\eps_1,\eps_2)$ half-stable so there is a rectangle $R$ which includes several consecutive $z_i$.
The Tester will estimate the rank $z$ of an element which is not quite $z_i$ but close.  By the half-stability property its frequency must be close to $g(z_{i+1})$ or to $g(z_{i-1})$.
It explains the $3$ conditions on the Tester.

\subsection{Tester with a reference frequency function $f$ }

 The streaming Algorithm~\ref{alg:Tester} consists of the following three steps in parallel for all $i=\{ 1,2,\cdots, \lceil \log_{1+\eps_1 }n\rceil \} $:\\
 
\begin{enumerate}
\item 
We sample each element of the universe uniformly with probability $1/a_i$ where
$a_i= \lceil\eps_1^2z_i/\log\log n \rceil $ and discard from stream $s$ all elements that are not in the sample. This defines a substream $s_i$, such that the element that is the $r$th most frequent in $s_i$ is likely to be close to the $z_i$th most frequent element in $s$. For $i$ large enough $r\sim \log\log n / \eps_1^2$.
%with probability $1/a_i$, where $$a_i=\lceil\eps_1^2.z_i/\log\log n \rceil$$ and discard from stream $s$ all elements that are not in the sample. This defines a substream $s_i$.  %The sample probability $1/a_i$ is chosen so that (assuming that the frequency distribution $g$ of the elements of stream $s$ equals $f$), in expectation substream $s_i$ contains $\Theta(1/\eps_1^2)$ elements whose number of occurences is greater than  $f(z_i)$. 
%Let $r=\lceil z_i/a_i\rceil $ denote the (rounded) expected number of elements of the sample whose number of occurences is greater than $f(z_i)$. If  $z_i\leq 1/\eps_1^3$, then $a_i=1$ for large $n$, and we only have one stream.

\item
On substream $s_i$, we  run the 
SpaceSaving algorithm with a table $T_i$ of size $K=O(\log n\cdot \log\log n)$ to approximately track the frequencies of the most frequent elements; 
Let $c_{i,r}$ denote the value of the $r$th largest counter in table $T_i$. 

\item 
The Tester accepts if, for most $i$'s ({\em most}, in a sense to be specified in the detailed algorithmic description), $c_{i,r}$ is close to $f(z_i)$. 
%A test for $i$ is {\it negative} if $c_{i,r}$ is not $3\eps_2$-close to $f(z_i)$. The tester rejects if there is a negative test with $z_i\leq 1/\eps_1^3$, or if there are  negative tests for 3 consecutive values of $i$ with $z_i> 1/\eps_1^3$. Otherwise it accepts.\\
\end{enumerate}

The complete description is spelled out in Algorithm~\ref{alg:Tester}. 
The algorithm would be exact if the $z_i$th most frequent element of $s$ belonged to $s_i$; if it had the $r$th largest counter in $T_i$; and if that counter was exactly equal to its frequency $f(z_i)$. Because of the noise in the sampling and in the SpaceSaving algorithm, none of those three properties is guaranteed, and they are sources of error. With additional assumptions on the input, we are able to bound the errors. \\

 \begin{notation}\label{n1}
The frequency function $g$ of the stream $s$ and 
the reference frequency $f$ are both from $\{1,2,..n\}$ to   $N$. 
 Both functions $f$ and $g$ are $(\gamma_1,\gamma_2)$-decreasing (see Definition~\ref{definition:decreasing}), $\eps_1,\eps_2$ are the Fr\'echet parameters (see Definition~\ref{definition:Frechet}), and $\delta$ is the desired error probability of the Tester (see Definition~\ref{definition:Tester}).
  Let $U$ be the set of elements $e$ and $\occ(e)$ is the number of occurences of $e$.  For each stream $s_i$, the counter $c_{i,r}$ of the Spacesaving algorithm is compared with  $f(z_i)$ where $r=\lceil  \frac{z_i}{ a_i } \rceil =\lceil  \frac{\log\log n }{ \eps_1^2} \rceil  $.  

 Let $K$ denote the size of the table used by Algorithm~\ref{alg:Tester}  for the substream $s_i$. We choose $K=O(\log n \cdot \log\log n)$. In Lemmas \ref{re1}, \ref{re}, it is precisely:
 $$K= 2\cdot\frac{z_i}{ a_i } \cdot \frac{(\gamma_1-1)}{1-\gamma_1/\gamma_2}\cdot (1+\eps_1^2)\cdot \frac{\log n}{\eps_2 \delta}=\frac{2 }{ \eps_2 . \eps_1^2} \cdot \frac{(\gamma_1-1)}{1-\gamma_1/\gamma_2}\cdot \frac{\log n .\log\log n }{\delta}\cdot(1+\eps_1^2)$$
 
 \end{notation}

%\begin{footnotesize}
{

\begin{algorithm}[ht]
{\bf Tester Algorithm  $A_1(\eps_1,\eps_2,\delta$; half-stable function $f$)}\\
 \KwData{a stream $s$ from a universe $\{ e_1, e_2, \ldots ,e_n \}$.  }

 ~~~~~~~~~{\bf Notations}: $z_i\gets \lceil(1+\eps_1^2)^i \rceil$ ;~~ $K\gets O(\log n\cdot \log\log n)$ \;
 
 ~~~~~~~~~~$a_i\gets \lceil\eps_1^2z_i/\log\log n \rceil $ ;~~ $h_i\gets$ uniform hash function $[1,n] \rightarrow [1,a_i]$ \;
 
  ~~~~~~~~~~$r \gets \lceil  \frac{z_i}{ a_i } \rceil =\lceil  \log\log n/\eps_1^2 \rceil$  ;
 
 % of size $n$. }
 %Compute the decomposition of $[1,n]$ into 
 %Boxes  Intervals according to Lemma~\ref{lemma:stepcompatible} for $f$.\\
{%TODO: passer en 10 points pour l' algo
 \For{each $i=1,2,\ldots ,\lceil \log_{1+\eps_1 }n\rceil $ in parallel:}
{
%$z_i\gets \lceil(1+\eps_1^2)^i \rceil$ ;~~ $K\gets O(\log n\cdot \log\log n)$ \;
%$K_i = K = 2\cdot\frac{z_i}{ a_i } \cdot \frac{(\gamma_1-1)}{1-\gamma_1/\gamma_2}\cdot (1+\eps_1^2)\cdot \frac{\log n}{\eps_2 \delta}$ \;
%If $z_i$ is not $\eps_1^2$-close to an Interval endpoint or  $z_i < 1/\eps_1^3$ then:\\
{\bf  1. Defining substreams} \;
% $a_i\gets \lceil\eps_1^2z_i/\log\log n \rceil $ ;~~ $h_i\gets$ uniform hash function over $[1,a_i]$ \;
 Let $s_i$ denote the substream consisting of those elements $e$ where $h_i(e)=1$ \;
{\bf  2. Dealing with substreams $s_i$ in parallel}  \;

	on substream $s_i$, run SpaceSaving with a table $T_i$ of size $K$\;
	 $c_{i,r}\gets$ the counter at position $r$ of  table $T_i$  \;
	
}

{\bf  Coherence Test} \;

 {

%\If{$\nexists (t,f(t))$ such that $(z_i,c_r)\simeq_{ (1\eps_1/3,2\eps_2/3)} (t,f(t))$}
\If{ $z_i \leq  \lceil 1/  \eps_1^3 \rceil  ~ \wedge ~c_{i,r}\not \simeq_{5.\eps_2} f(z_i)$ }
	 {output NO}
}
\If{ $\exists i ~~~z_i >\lceil 1/ \eps_1^3 \rceil  ~ \wedge \\ ~c_{i,r} \not \simeq_{{5.\eps_2}} f(z_i)~ 
\wedge 
 ~c_{i,r} \not \simeq_{5.\eps_2} f(z_{i-1})~ \wedge ~c_{i,r} \not \simeq_{5.\eps_2} f(z_{i+1})$ }
	 {output NO}
{output YES}\\
%\If{ $\exists i ~~~z_i >\lceil 1/ \eps_1^3 \rceil  ~ \wedge ~c_{i,r} \not \simeq_{3.\eps_2} f(z_i)~ 
%\wedge \\
% ~c_{i+1,r} \not \simeq_{3.\eps_2} f(z_{i+1})~ \wedge ~c_{i+2,r} \not \simeq_{3.\eps_2} f(z_{i+2})$ }
%	 {output NO}
%{output YES}\\
}

 \caption{A Streaming Tester to check whether the function $g$ of frequencies of elements arriving in a stream matches a given frequency function $f$}\label{alg:Tester}
\end{algorithm}

}
%\end{footnotesize}

If  $z_i \leq  \lceil 1/  \eps_1^3 \rceil $, then $a_i = \lceil\eps_1^2z_i/\log\log n \rceil  \leq \lceil  1/\eps_1.\log\log n \rceil=1$ there is only one stream  and we query the same Spacesaving table in position $r= \lceil  z_i/a_i \rceil=\lceil  z_i \rceil$. We start distinct substreams for $a_i\geq 2$ and then query the position $r= \lceil  z_i/a_i \rceil$ which is independent of $i$.

\subsubsection{Analyzing the Tester Algorithm \ref{alg:Tester}}

What does this algorithm accomplish?  We use the Property Testing framework with the relative Fr\' echet distance.
 The notion of a Property Tester goes back to \cite{B93} and the streaming version to \cite{F02}. 
 
\begin{definition}\label{definition:Tester}
Let $\eps_1,\eps_2,\delta \in (0,1)$. A {\em streaming Tester for a  frequency function $f$} is a streaming algorithm $A$ which, given   a %non-increasing $ (\eps_1,\eps_2)$-half-stable 
function $f$ over $\{ 1,2,\cdots ,n\}$,  takes as input  a stream of elements from a universe of size $n$ defining a frequency function $g$ such that $g(j)$ is the number of occurrences of the $j$th most frequent element in the stream and:  %, assumed to be half-stable;
\begin{itemize}
\item  if $f=g$   then $A$ accepts with probability  at least $1-\delta$; and 
\item if $g$ is $ ( \eps_1, \eps_2)$-far from $f$ for the relative Fr{\'e}chet distance  then $A$ rejects with probability at least $1-\delta$.
\end{itemize}
\end{definition}

%A more general {\em Tolerant  Tester} replaces the first condition with the tolerant version:   if $g$ is $(\eps_1/10,\eps_2/10)$-close to $f$ for the relative Fr\'echet distance  then $A$ accepts with probability  at least $1-\delta$.
In this definition, the gap  between  $(0,0)$ and $ ( \eps_1, \eps_2)$ assumes the parameters $\eps_1, \eps_2$ are fixed  in the analysis.
We want Algorithm~\ref{alg:Tester} to be a streaming Tester for $f$. 
%For that, we need two assumptions on the frequency distributions being tested: they must be {\emph half-stable} and $(\gamma_1,\gamma_2)$-decreasing, two notions that we now define.\\
%\begin{theorem-non}[Separation theorem] \label{fml}
%If $f$  and $g$ are $(3\eps_1, \eps_2)$-half-stable and $f \not \sim_{(3\eps_1, 3\eps_2)} g$ then there exists an  $(\eps_1, \eps_2)$-rectangle which separates $f$ and $g$.
%\end{theorem-non}
%\begin{definition}($(\gamma_1,\gamma_2)$-decreasing) Let $\gamma >1$. 
%A non-increasing function $f$ with domain $\{ 1,\cdots ,n\}$
% is $(\gamma_1,\gamma_2)$-decreasing if for all $t$ such that $1 \leq \gamma. t \leq n$:
%$f(\lceil \gamma.t \rceil )\leq f(t)/2$.
%\label{definition:decreasing}
%\end{definition}
We assume that the frequency functions are  $(\eps_1,\eps_2)$-half-stable  (definition \ref{hs}) and  $(\gamma_1,\gamma_2)$-decreasing (definition \ref{definition:decreasing}).
We can now state the main result:\\

%First appearance of the main theorem
\begin{theorem-non}
  \label{main}
  \thmtext
\end{theorem-non}

We generalize the Tester to a tolerant version in Corollary \ref{tt} where we distinguish between
$f$ and $g$ are $ ( \eps_1, \eps_2)$-close from the case where $f$ and $g$ are $ ( \eps'_1, \eps'_2)$-far with $\eps'_i > \eps_i$.

\subsection{Tester to compare two streams}\label{2s}
Assume two streams $s_1$ and $s_2$ of elements on two different domains of the same size $n$.  The algorithm $A_2$  keeps two sets of tables  $T_i^1$ for the stream $s_1$ and $T_i^2$ for the stream $s_2$, for $i=1,2,\ldots ,\lceil \log_{1+\eps_1 }n\rceil $.\\

We use the same definition \ref{definition:Tester}, with different parameters $c_1, c_2$. We modify the previous Tester $A_1$ by replacing the reference frequency function $f$ by the approximation $ \bar{g_1}$ of $g_1$ defined in section \ref{2streams}. The step function $ \bar{g_1}(t)  $ is defined by the counters
$c_{i,r}^1$. Because of errors, the $c_{i,r}^1$ are not necessarly monotone, so we define:
$$\bar{c}_{i,r}^1=    Min _{j\leq i }\{ c_{j,r}^1 \}  $$
the monotone version of $c_{i,r}^1$. In the coherence test we will compare $\bar{c}_{i,r}^1$ with $c_{i,r}^2$ as
$ \bar{g_1}(t)  $ is the step function defined by $\bar{c}_{i,r}^1$ and replaces $f$ in the Tester $A_1$.\\

%\begin{definition}\label{test2s}Let $\eps_1,\eps_2,\delta \in (0,1)$. A {\em streaming Tester} is a streaming decision algorithm $A$ which  takes two  input  streams  defining   two frequency distributions $g_1$ and $g_2$:\begin{itemize}\item  if $g_1 \simeq_{\eps_1,\eps_2}g_2$   then $A$ accepts with probability  at least $1-\delta$; and \item if $g_1$ and $g_2$ are  $ (10\eps_1,10\eps_2)$-far  for the relative Fr{\'e}chet distance then $A$ rejects with probability at least $1-4\delta$.\end{itemize}\end{definition}

{
\begin{algorithm}[ht]
{\bf Tester Algorithm A$_2(\eps_1,\eps_2,\delta$)}\\
 \KwData{Two streams $s_1$  and $s_2$ }
 
 ~~~~~~~~~{\bf Notations}: $z_i\gets \lceil(1+\eps_1^2)^i \rceil$ ;~~ $K_1\gets O(\log n\cdot \log\log n)$ \;
 
 ~~~~~~~~~~$K_2\gets O(\log n\cdot \log\log n)$ ;~~$a_i\gets \lceil\eps_1^2z_i/\log\log n \rceil $ ;
 
   ~~~~~~~~~~$h_i\gets$ uniform hash function $[1,n] \rightarrow [1,a_i]$ ;~~
$r \gets \lceil  \log\log n/\eps_1^2 \rceil$  ;

 % of size $n$. }
 %Compute the decomposition of $[1,n]$ into 
 %Boxes  Intervals according to Lemma~\ref{lemma:stepcompatible} for $f$.\\
{%TODO: passer en 10 points pour l' algo
 \For{each $i=1,2,\ldots ,\lceil \log_{1+\eps_1 }n\rceil $ in parallel:}
{
%$z_i\gets \lceil(1+\eps_1^2)^i \rceil$ ;~~ $K\gets O(\log n\cdot \log\log n)$ \;
%$K_i = K = 2\cdot\frac{z_i}{ a_i } \cdot \frac{(\gamma_1-1)}{1-\gamma_1/\gamma_2}\cdot (1+\eps_1^2)\cdot \frac{\log n}{\eps_2 \delta}$ \;
%If $z_i$ is not $\eps_1^2$-close to an Interval endpoint or  $z_i < 1/\eps_1^3$ then:\\
{\bf  1. Defining substreams} \;
% $a_i\gets \lceil\eps_1^2z_i/\log\log n \rceil $ ;~~ $h_i\gets$ uniform hash function over $[1,a_i]$ \;
 Let $s_i^1$ (resp. $s_i^2$) denote the substream consisting of those elements $e$ of $s_1$ (resp. $s_2$) where $h_i(e)=1$ \;
{\bf  2. Dealing with substreams $s_i$ in parallel}  \;

	on substream $s_i^1 $ (resp. $s_i^2$), run SpaceSaving with a table $T_i^1$ (resp. $T_i^2$) of size $K$\;
	
%{\bf  3. Coherence Test} \;

~~$c_{i,r}^1\gets$ the counter at position $r$ of  table $T_i^1$  \;
~~$c_{i,r}^2\gets$ the counter at position $r$ of  table $T_i^2$  \;

%\If{$\nexists (t,f(t))$ such that $(z_i,c_r)\simeq_{ (1\eps_1/3,2\eps_2/3)} (t,f(t))$}

}
{\bf  Coherence Test} \;
\If{ $z_i \leq  \lceil 1/  \eps_1^3 \rceil  ~ \wedge ~\bar{c}_{i,r}^1\not \simeq_{5.\eps_2}c_{i,r}^2 $}
	 {output NO}

\If{ $\exists i ~~~z_i >\lceil 1/ \eps_1^3 \rceil  ~ \wedge ~\bar{c}_{i,r}^1\not \simeq_{5.\eps_2}c_{i,r}^2 ~ \wedge ~\bar{c}_{i-1,r}^1\not \simeq_{5.\eps_2}c_{i-1,r}^2 ~ \wedge ~\bar{c}_{i+1,r}^1\not \simeq_{5.\eps_2}c_{i+1,r}^2 $ }
	 {output NO}
{output YES}
}

 \caption{Streaming Tester to compare two  streams}\label{alg:Tester2}
\end{algorithm}
}

\begin{theorem-non}
  \label{main3}
  \thmtextc
\end{theorem-non}
\subsection{Tester to compare two marginal distributions}\label{2m}
If $f$ is unknown and $g_1$ and $g_2$ are two marginal distributions defined from $g$ where each element $e_i$ is a $d$-tuple $(e_{i,1},...e_{i,d})$ and two projections $\pi_1, \pi_2$ selecting two subtuples. The marginal frequency distributions $g_1$ and $g_2$, defined similarly as in the previous sections on these subtuples, which may have different supports  but approximate sizes. In this case we take the larger domain as a reference and complete the smaller domain by new points with frequency $0$.

We present $A_3$ 
%\ref{alg:Tester3} 
which directly calls $A_2$.\\

{
\begin{algorithm}[ht]
{\bf Tester Algorithm  A$_3(\eps_1,\eps_2,\delta,\pi_1,\pi_2$)}\\
 \KwData{a stream $s$ of $d$-tuples from a universe and two projections $\pi_1,\pi_2$  }
 
 {\bf  1. Defining substreams} \;
 %$a_i\gets \lceil\eps_1^2z_i/\log\log n \rceil $ ;~~ $h_i\gets$ uniform hash function over $[1,a_i]$ \;
 Let $s^1$ (resp. $s^2$) denotes the substream consisting of those elements $\pi_1(e)$ (resp.   $\pi_2(e)$) where $h_i(e)=1$\;
 
 {\bf  2. Apply A$_2(\eps_1,\eps_2,\delta)$ on  $s^1$ and $s^2$ \;

}
 \caption{Streaming Tester to compare two marginal distributions}\label{alg:Tester3}
\end{algorithm}
}

%If $f$ is unknown, we use the more general definition \ref{test}:
\begin{corollary}\label{st}

%\begin{T2}
\thmtextb
\end{corollary} 

%\begin{theorem-non}  \label{main2}\thmtextb\end{theorem-non}

%The setting for  theorem \ref{main2} is in the appendix \ref{comp}.

\section{General space lower bounds}

 We first consider the case when $f$ is uniform. We then assume that $f$ is  $(\gamma_1,\gamma_2)$-decreasing  and want to test the half-stability property.

\subsection{A lower bound when $f$ is uniform}

A classical observation is that in the worst-case, the approximation of $F_{\infty}={\rm Max}_j ~f(j)$ requires space $\Omega(n)$, using a standard reduction from Communication Complexity.  In \cite{KPW21}  the Unique-Disjointness problem for $x,y \in\{0,1\}^n$ is reduced  to the approximation of $F_{\infty}$ on a stream $s$. Another standard problem which requires space $\Omega(n)$ for the One-way Communication complexity is the Index$(x,y)$ problem, see \cite{KN96}, where $x\in \{0,1\}^n$, $y\in \{1,2,...n\}$ and the goal is to compute $x_y\in \{0,1\}$. We write
Index$(x,y)=x_y$, as Alice holds $x$ of length $n$, Bob holds $y$ of length $\log n$ and only Alice can send information to Bob. Notice that we can assume that $| \{i : ~x_i=1\}  | =O(n)$ for example $n/2$, otherwise Alice would directly send these positions to Bob.

We show in the next result a simple reduction from the Index problem to the {\em Tolerant streaming Test problem}
which given $f$ and a stream $s$ over the items $a_1,...a_n$, which defines a frequency $g$,  decides:
either $f \sim_{\eps_1,\eps_2} g$  or $f \not\sim_{\eps'_1,\eps'_2} g$  with h.p.\\

\begin{theorem-non}\label{lb}
The Tolerant streaming Test problem  requires space $\Omega(n)$ when $f$ is uniform over its support.

\end{theorem-non}
\begin{proof}
Consider the  following reduction from Index to the streaming Test problem. Given
 $x \in\{0,1\}^n$ and $y\in \{1,2,...n\}$ the inputs to Index, let $f$ be the frequency distribution on a domain $\{1,2,...n\}$  for a stream on the elements $\{a_1,a_2,...a_n\}$ and let 
 $f(i)=1$ for $i=1,...k=O(n)$ if  $x_i=1$.
The stream $s$ is determined by the elements of $x$ of weight $1$, followed by the element
$a_y$ associated with $y$, i.e.
$a_{i_1},...a_{i_k}$ where  $x_{i_j}=1$,  followed by $a_y$. It is of length $k+1$.

If   Index$(x,y)=1$ then  frequency  $g$ of the stream has an element of frequency $2$. The point $(1, 1)$ of $f$ is far from the closest point
$(1, 2)$ of $g$. Hence  $f \not\sim_{10\eps,10\eps} g$. 

If   Index$(x,y)=0$ then $g$ is uniform over $k+1$ elements. The points 
$(i, 1)$ of $f$ for $i=1,2...k$ are at relative distance $(0,0)$ 
 from the closest point of $g$ for $i=1,2...k$. The point 
$(k+1, 1)$ of $g$ is at relative distance $(\eps,0)$ from the point $(k,1)$ of $f$.
Hence  $f \sim_{\eps,\eps} g$ for $\eps=1/k$ and $n$ large enough.

We reduced a Yes-instance to Index to a No-instance of Test, and a No-instance of Index to a Yes-instance of Test.

As Index requires space $\Omega(n)$, so does the  streaming Test problem.
\end{proof}

\subsection{A lower bound for the half-stability property when  $f$ is $(\gamma_1,\gamma_2)$-decreasing}

In this section, we assume without loss of generality that $f$ is $(\gamma_1,\gamma_2)$-decreasing for $\gamma_1=\gamma_2=2$ and we reduce the Index problem between Alice and Bob to the decision whether $g$ is half-stable or not.

Given a binary word $x$ of length $n'$ given to Alice and a value $i \in \{1,2,...n'\}$ given to Bob, let $n=2n'$. We construct a stream  on an alphabet of size $3n/4$. The elements are $\{a_1,...a_n'\}$ associated to the input string of length $n'$ and $\{b_1,...,b_{n'/2}\}$ to make sure that $f$ is $(\gamma_1,\gamma_2)$-decreasing.

Assume $\gamma=2$ and let $f_0$ be the step function described in figure \ref{stepf}: $f(t)=2^{k-i}$: for $2^{i-1}\leq t <2^i$. If $n'=8$, let $n=2n'=16=2^k$, hence $k=4$.

The stream starts with a prefix using the $b_i$ followed with a suffix using the $a_j$.

 \begin{figure}[ht]
\includegraphics[height=70mm]{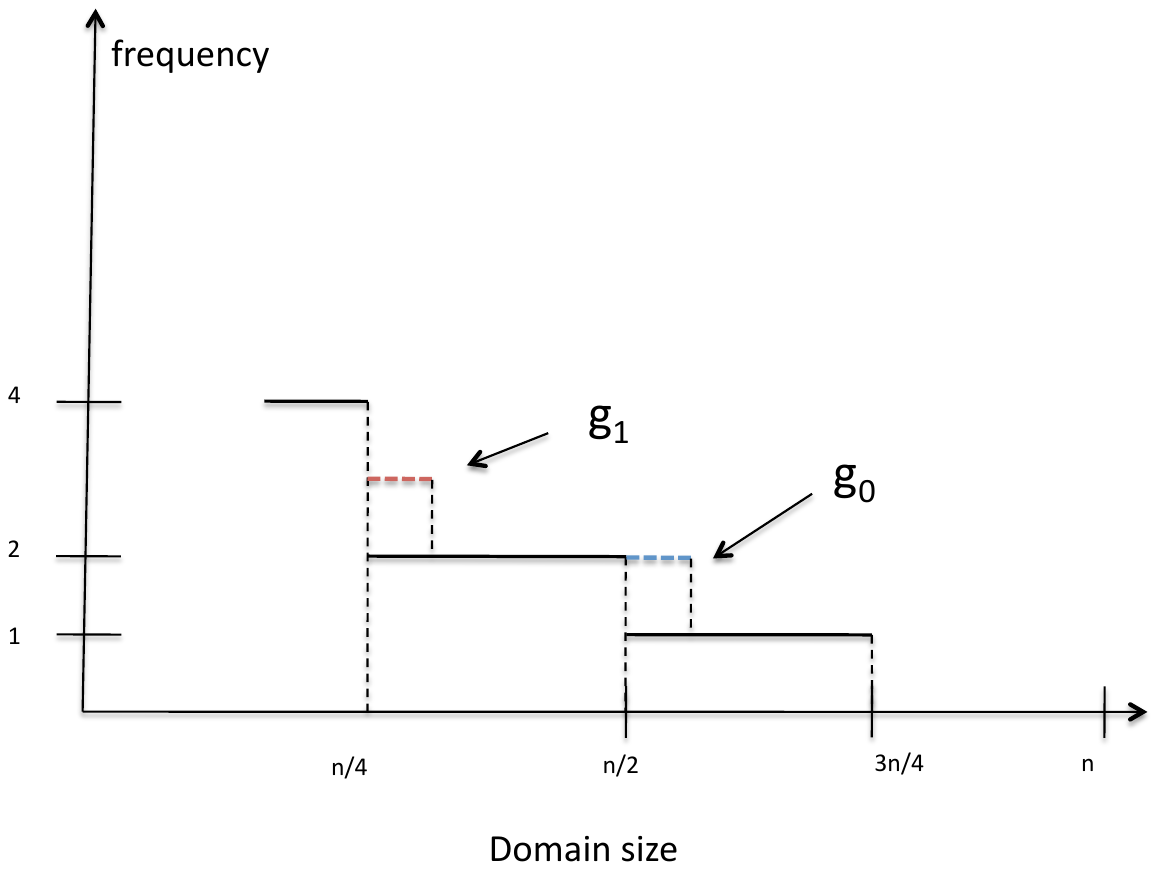}
\caption{Step function $f_0$ and frequency functions $g_0$ and $g_1$ used in theorem \ref{lb2}} 
\label{stepf}
\end{figure}

%Let $f_0$ be a double step function: $$(1,3), (2,3),...(k,3),(k+1,2),(k+2,1),(k+3,1),...(n,1)$$

\begin{theorem-non}\label{lb2}
The half-stability property  requires space $\Omega(n)$ when $f$ is $(\gamma_1,\gamma_2)$-decreasing.
\end{theorem-non}
\begin{proof}
Consider the  following reduction from Index to half-stability. Given
 $x \in\{0,1\}^n$ and $y\in \{1,2,...n\}$ the inputs to Index, let $f$ be the $f_0$ function described in figure \ref{stepf}. 

The stream $s$ is determined by a prefix using the elements $\{b_1,...,b_{n'/2}\}$ , followed by the elements of $x$ of weight $1$ repeated twice, followed by the elements of $x$ of weight $0$,  followed by the element
$a_y$ associated with $y$.  The prefix guarantees that  $g$ may be close to a $(\gamma_1,\gamma_2)$-decreasing function and the frequencies of the $b_j$ follow the function $f_0$.  After the prefix, the stream is  
$a_{i_1}^2,...a_{i_k}^2$ where  $x_{i_j}=1$ where $k=O(n)$,  followed by $a_{i_1},...a_{i_k}$ where  $x_{i_j}=0$, followed by $a_y$. It is of length $n'+1$.

If   Index$(x,y)=1$ then  the  frequency  $g_1$ of the stream has an element of frequency $3$, as $g(n/4+1)=3$, as $a_y$ is one of the elements $a_{i_j}$ which appears as $a_{i_j}^2$. The point $(n/4, 4)$ of $f$ and the  point $(n/4+2,2)$ of $f$ are $\eps$ far from the point
$(n/4+1, 3)$ of $g_1$, for $\eps> 1/2$. Hence  $f \not\sim_{10\eps,10\eps} g_1$  for $\eps <1/20$.

If   Index$(x,y)=0$ then the frequency  $g_0$ of the stream  satisfies $f \sim_{\eps,\eps} g_0$, because $g_0(n/2 +1)=2$ and the point $(n/2, 2)$ of $f$ is 
$(\eps,\eps)$ close to the point $(n/2 +1, 2)$ of $g_0$.
We reduced a Yes-instance of the  Index problem to a No-instance of Test, and a No-instance of Index to a Yes-instance of Test.

As Index requires space $\Omega(n)$, so does the  streaming Test problem.
\end{proof}
As an example, suppose $n'=8$ and $x=01100110$, $k=n'/2=4$, $y=4$  hence $x_y=0$.
The stream is:
$$s= b_1^{16}.b_2^{8}.b_3^{4}.b_4^{4}. a_2^2 a_3^2 a_6^2 a_7^2.a_1a_4a_5a_8.a_4$$
In this case $f \sim_{\eps,\eps} g$. If $y=3$, we replace the last item by $a_3$ and the frequency of $a_3$ is $3$, hence $g$ has a double jump and is far from $f$.

 \section{The Spacesaving algorithms and a  relative bound on the error}

Given a stream of $N$ elements drawn from a universe $U$ of size $n$, let $f(j)$ denote the frequency (number of occurences) of the $j$th most frequent element, so that  $ f(1)\geq f(2)\geq \cdots \geq f(n) \geq 0$ and $\sum_{i=1}^n f(j)=N$.

%For example, in the case of a graph given as a stream of $m$ edges,  {\em i.e.} a stream of pairs of vertices, we can define the elements of the stream as the vertices, so the length of the stream is $N=2m$, and $(f(j))$ is the degree sequence of the graph.

We are particularly interested in frequencies which have a compact representation. 
For example, 
{\em uniform frequencies} where $f(i)= N/n$, 
 {\em Zipf} frequencies (also called heavy-tailed, or scale-free, or power-law)  with parameter $\alpha$, where $f(i)= cN/i^\alpha$ with $c=1/\sum_{1\leq j\leq n}( 1/j^{\alpha} )$, and  
 {\em  geometric} frequencies where $f(i)=cN/2^{i}$ with $c=1/\sum_{1\leq j\leq n} 1/2^j$.

For Zipf frequencies with parameter $\alpha$ the maximum frequency is $f(1)=\Theta(N)$ if $\alpha >1$ and $f(1)=\Theta(N/\log n)$ if $\alpha=1$. 

\subsection{The Spacesaving algorithm}\label{rb}

The SpaceSaving algorithm introduced in~\cite{MAE2005} computes an approximation of the frequencies of the $k$ most frequent items elements in a stream for the {\em insertion only} model and an additive error. It uses a table $T$ of triplets $T[j]=(e, c_j, \eps_j)$ where $e \in U$ is an element of the universe $U=\{e_1,e_2,\cdots ,e_n\}$, $c_j \in N$ is a counter approximating the number of occurences of element $e$ and $\eps_j \in N$, $\eps_j< c_j$ is a bound on the error between the counter and the correct number of occurences of $e$ in the stream. The table, of size $K$, is ordered by counters: $c_1 \geq c_2,...\geq c_K$. Assume $k<K$. \\

\begin{algorithm}[H]\label{topk}
{\bf Algorithm Top-k$(k, K)$}\\
 \KwData{a stream $s$ of length $N$, from a universe $U=\{ e_1,e_2,\ldots ,e_n\}$.}% of size $n$. }

 $T[j]\gets (-,0,0)$ for every  $j\in[1,K]$\;
 \While{stream $S$ is flowing}{
  read next element $e$ of $S$\;
  \eIf{$e$ is in the table $T$ at position $j$ }{
   increment $c_j$\;
   }{
%   let $j$ be a position in $T$ s.t. $c_j=\min(c_1,c_2,\ldots ,c_K)$\; 
   Replace $T[K]=(e', c_K, \eps_K)$ by 
   $T[K]=(e, c_K+1, c_K)$ \;
   Reorder $T$ by non-increasing values of $c_j$ \;
  }
 }
  \KwResult{the  sequence $S$ of the first $k$ elements}
  \vspace{5mm}

 \caption{The Top-k algorithm}
\end{algorithm}

\vspace{5mm}

\begin{notation}
For a table position $j\in [1,K]$, let  $\tilde{e_j}$ denote the element with the $j$th largest count in the table.
We define $\sigma(j)=i$ if         $\tilde{e_j}$ has frequency $f(i)$.
\end{notation}
%Let $\sigma: [1,K] \rightarrow A$ map position $j\in [1,K]$ to the element $e\in A$ such that $T[j]=(e,c_j,\eps_j)$.
%Let $T=\{(e_{\sigma(1)}, c_1, \eps_1),...(e_{\sigma(K)}, c_K, \eps_K)\}$ with 
%$\sigma: [1,K] \rightarrow [1,n]$ and the counters $c_1 \geq c_2,...\geq c_K$. 
Thus $f(\sigma(i))$ is the frequency associated with the element whose counter is  $c_i$.  Algorithm Top-k guarantees that $\sigma $ is injective. 
In the ideal case in which $\sigma(i)=i$ for all $i\in [1,K]$, then $T $ contains the $K$ most frequent elements of $U$,  ordered by non-increasing frequency.
Algorithm Top-k satisfies the following properties:
\begin{enumerate}
\item
$\sum_{1\leq j\leq K} c_j=N$ \label{prop:1}
\item
For all $j\leq K $, $\epsilon_j\leq c_K$.\label{prop:2}
\item
For all $j\leq K $, $c_j - \epsilon_j \leq f(\sigma(j)) \leq c_j$.\label{prop:3}
\item 
For each element $e\in U$ not in $T$, i.e. for any index $i\notin Im(\sigma)$:
$f(i) \leq c_K$.\label{prop:4}
\end{enumerate}

The size $K$ of the table  can be tuned to provide  the approximate Top-k elements or the exact Top-k elements, in some special cases.
Let $S^*$ be the set of top $k$ most frequent elements.
The following Lemma is implicitly present in~\cite{MAE2005}. 

\begin{lemma}\label{exss}(adapted from \cite{MAE2005})
\begin{enumerate}
\item
(Exact result) If  $c_K\leq f(k) -f(k+1)$, then the  Top-k algorithm gives the  exact solution $S^*$.
\item
(Approximate result) If  $c_K\leq \eps.f(k)$, then $S$ contains every element $e_i$ such that 
$f(i) \geq (1+\eps).f(k) $ and no element $e_i$ such that  $f(i) \leq (1-\eps).f(k)$.
\end{enumerate}
\end{lemma}

%\iffalse
\begin{proof}
Assume $c_K\leq f(k) -f(k+1)$. From property~\ref{prop:4}, if $e\in U$ is not in the table $T$, its frequency 
$f(i) \leq c_K$. As $c_K\leq f(k) -f(k+1) < f(k)$, hence $f(i) < f(k)$ and $e \notin S^*$. Let us show that if $e\in T-S$, then $e \notin S^*$.

Let $i,j$ two elements of $T$ such that $f(i) > f(j) +c_K$. The corresponding counters
$c_{\sigma^{-1}(i)}$ and $c_{\sigma^{-1}(j)}$ are in the right order, i.e.
$$c_{\sigma^{-1}(i) }>c_{\sigma^{-1}(j)}$$
Apply properties~\ref{prop:2} and~\ref{prop:3}:

$$c_{\sigma^{-1}(i)} \geq f(i) > f(j) +c_K \geq f(j) +\eps_j \geq c_{\sigma^{-1}(j)}$$

If $i\in \{1,2,...k\}$ and $j \notin \{1,2,...k\}$, then:
$$f(i)- f(j) > f(k) -f(k+1) \geq  c_K $$
Hence the counters $c_{\sigma^{-1}(1)},....c_{\sigma^{-1}(k)}$ are all greater than the counters $c_{\sigma^{-1}(j)}$ for $j>k$. Hence $S=S^*$.\\

Assume $c_K\leq \eps.f(k)$, the figure \ref{c-f} shows that if $f(i) < (1-\eps)f(k)$ then
$c_{\sigma^{-1}(i)}$ is smaller than all the counters of elements of $S^*$, hence 
$i \notin S$. If $f(j) >(1+\eps)f(k)$, then 
$c_{\sigma^{-1}(j)}$ is larger than all the counters of elements of
$U-S^*$, hence $i \in S$.
\end{proof}
%\fi

 \begin{figure}[h]
\includegraphics[width=\linewidth]{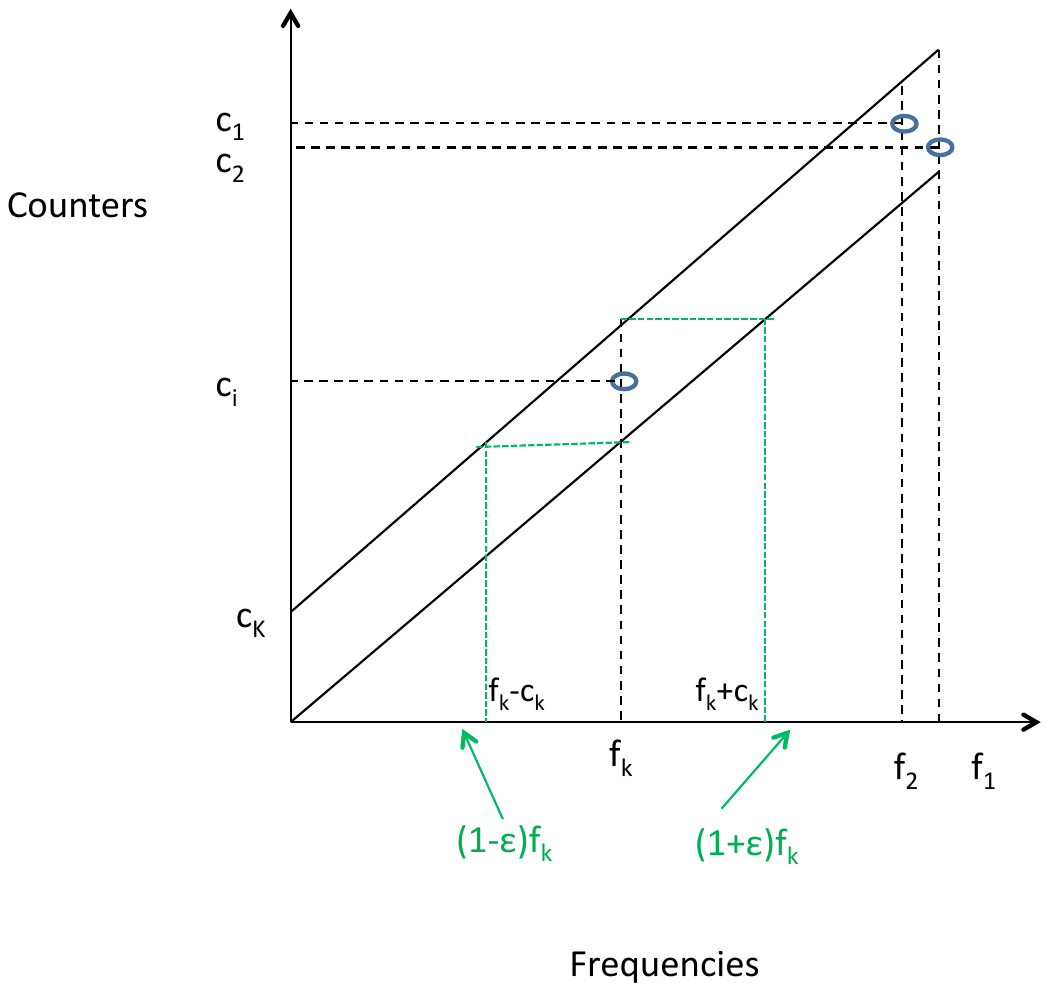}
\caption{Frequencies-Counters relation:  for each $1\leq k\leq n$, the   $i$-th element of the table $T$ is $(e,c_i,\eps_i)$ where $c_i$ is the $i$-th counter and
$\sigma(i)=k$. Then   $  f(k) \leq c_i   \leq f(k)+ \eps_k \leq f(k) +c_K$. By properties~\ref{prop:2} and~\ref{prop:3} the points $(f(k), c_{\sigma^{-1}(k)})$ are above the diagonal and below the diagonal shifted by $c_K$.}

%$ c_{\sigma^{-1}(k)}  -c_K \leq c_{\sigma^{-1}(k)}  -\eps_k     \leq f(k) \leq c_{\sigma^{-1}(k) }$.}
\label{c-f}
\end{figure}

When the table $T$ of size $ K > k$  is such that:
\begin{equation}
c_K \leq f(k) - f(k+1)   ~~~~ \tag{2}\label{e1}
\end{equation}
Lemma \ref{exss} for  the  condition \eqref{e1} guarantees that the Top-k algorithm gives an exact solution. In fact the  $k$ first  elements of the table T are in the right order.
In the original paper\cite{MAE2005} , Lemma~\ref{sb} bounded the additive error $c_K$ using a simple averaging argument.

Notice that  $c_K \leq N/K$ by the uniform bound on the minimum value, hence $K=O(\frac{1}{\eps} )$, then $c_K \leq \eps.N$ and the frequency $f(i)$ can be approximated with an additive error less than $\eps.N$. 
In  \cite{B10} a better bound is given, which is a lower bound in the worst-case. If we want a relative error for the Top-k algorithm, i.e the hypothesis $c_K \leq \eps.f(k)$ of lemma \ref{exss}, we need to use the $(\gamma_1,\gamma_2)$-decreasing hypothesis, in sections \ref{tkre} and \ref{ssaving}.
We can summarize the various previous additive bounds in table \ref{tab}.

 \begin{table}
 \caption{Error bounds for the Top $k$ elements, $i\leq k$, $K=O(k / \eps)$}
 \label{tab}
 %\(
%\begin{array}[b]{  | l l l  l |   }
\begin{tabular}[b]{  |l l | l |l }
\hline
{\rm\bf Top ~k~elements ~for~ } $K=O(k / \eps )$  & {\rm\bf Error ~bound}\\
\hline
%\hline	

{\rm SpaceSaving}\cite{MAE2005}    &$ |f(i)-c_i | \leq 2\eps.N$\\
\hline
{\rm SpaceSaving ~with ~strong ~error~ Bounds}\cite{B10}  & $|f(i)-c_i | \leq \frac{\eps }{ k }. F^{res(k)}$\\
\hline

{\rm SpaceSaving ~for ~}$(\gamma_1/\gamma_2)$-{\rm decreasing~functions}  & $|f(i)-c_i | \leq \eps.f(k).\frac{(\gamma_1-1)}{1-\gamma_1/\gamma_2}\leq \eps.f(i).\frac{(\gamma_1-1)}{1-\gamma_1/\gamma_2}$\\
\hline
%\end{array}
\end{tabular}
%\)

\end{table}

\subsection{A tighter analysis of the SpaceSaving algorithm }\label{ssa}

Here, we prove the following improvement to Lemma~\ref{sb}'s bound on $c_K$, which is a special case  when $u=0$. In \cite{B10}, a specific stream shows that the bound is tight.

%Consider the normalized distribution of frequencies, denoted by $\overline{f}_i=\frac{f(i)}{N}$ and the normalized counters $\overline{c}_i=\frac{c_i}{N}$. 
Consider the  cumulative distribution of frequencies, denoted by $F_t=\sum_{1\leq i\leq t} f(i)$ and $F_0=0$ and the residual cumulative distribution of  frequencies 
$$F^{res(t)}=\sum_{t+1\leq i\leq n} f(i)$$

\begin{lemma}\label{ss}
% (inspired from~\cite{B10})
Let $K$ denote the size of the table and $c_K$ be an integer as defined in the Space Saving algorithm. Then $$c_K \leq \min_{u<K/2} \frac{F^{res(u)}}{K-2u}$$
\end{lemma}
\begin{proof}
Let $u$ be an integer in interval $[1,K/2[$.  To prove the Lemma, it suffices to argue that $c_K \leq \frac{F^{res(u)}}{K-2u}$.
The minimum value $c_K$ of the counters  is less than the average of the counters over the interval $[u+1,K]$, so (using properties~\ref{prop:1} and~\ref{prop:3}): 
$$c_K \leq
 \frac{\sum_{j=u+1}^K c_j}{K-u} = 
\frac{N-\sum_{j=1}^u c_j}{K-u}  \leq
\frac{N-\sum_{j=1}^u f(\sigma(j))}{K-u}.$$

Notice  that $f(\sigma(1))+f(\sigma(2))+...f(\sigma(u)) \leq  f(1)+f(2)+...f(u)$ because  $(f(j))$ is a non-increasing sequence. 

Let us prove that for each $i$,  $$f(\sigma(i)) \geq f(i) - c_K$$

For every $j\in [1,i]$ we have $f(j)\leq c_{\sigma^{-1}(j)}$ by Property~\ref{prop:3}. Hence 
$f(i)=\min_{1\leq j\leq i} f(j)\
\leq \min_{1\leq j\leq i} c_{\sigma^{-1}(j)}$. 
But by definition of $c_i$, 
$\min_{1\leq j\leq i} c_{\sigma^{-1}(j)}\leq c_i$, 
and by Property~\ref{prop:3} again, $c_i\leq f(\sigma(i))+c_K$. Therefore $f(\sigma(i)) \geq f(i) - c_K$. 

Hence:
 $f(1)- c_K +f(2)-c_K+...f(u) -c_K\leq  f(\sigma(1))+f(\sigma(2))+...f(\sigma(u))$ and:
 $$N-\sum_{j=1}^u f(\sigma(j)) \leq N-\sum_{j=1}^u f(j) +u.c_K$$
 
$$c_K \leq \frac{N-\sum_{j=1}^u f(\sigma(j))}{K-u}\leq  \frac{N-\sum_{j=1}^u f(j) +u.c_K}{K-u}= \frac{N-\sum_{j=1}^u f(j) }{K-u}+\frac{u.c_K}{K-u}$$

$$ c_K .\frac{K-2u}{K-u}\leq \frac{N-\sum_{j=1}^u f(j) }{K-u}=\frac{\sum_{j=u+1}^n f(j)}{K-u} =
\frac{F^{res(u)}}{K-u}$$
If $u <K/2$ then $K-2u>0$:

$$ c_K \leq \frac{F^{res(u)}}{K-2u}$$
  
As this true for all $u <K/2$, then
$$c_K \leq \min_{u<K/2} \frac{F^{res(u)}}{K-2u}$$
\end{proof}

\subsubsection{Application to  Zipf distributions}

Assume a Zipf distribution of parameter $\alpha >1$:  $f(i)=cN/i^{\alpha}$, where $c=1/(\sum_{i=1}^n 1/i^{\alpha})$.  We apply Lemma \ref{ss} to upper bound the main uncertainty parameter $c_K$.

\begin{lemma}\label{zb}
 Let $K$ denote the size of the table.    Then:
$$c_K\leq \frac{\Theta(N)}{K^{\alpha}}.$$
\end{lemma}

\begin{proof}
Since $\alpha >1$, we have $c=\Theta(1)$ as $n  \rightarrow \infty$. 
$$F^{res(u)}=c N.\sum_{i=u+1}^{n}   \frac{1}{i^{\alpha}} $$
$$\int_{u+1}^{n}  \frac{dx}{x^{\alpha}}   \leq  \sum_{i=u+1}^{n}  \frac{1}{i^{\alpha}}\leq \int_{u}^{n}  \frac{dx}{x^{\alpha}}  $$

$$F^{res(u)} \leq \frac{\Theta(N)}{u^{\alpha-1}}$$
By lemma \ref{ss}, for $u <K/2$,  $c_K \leq \frac{F^{res(u)}}{K-2u}$. Hence for $u=K/3$:
$$c_K \leq \frac{F^{res(K/3)}}{K/3}\leq \frac{\Theta(N)}{K^{\alpha}}$$
\end{proof}

We need to analyse the size of $K$ in the Top-k algorithm \ref{topk} as a function of $k$ for Zipf distributions.
%The    frequencies of the Zipf distribution with parameter $\alpha$ are ${f}_i\sim\frac{c N}{i^{\alpha}}$, hence $\overline{F}_u\sim1-\frac{c'}{u^{\alpha-1}}$.  In this case:

$$f(k) - f(k+1)=cN.(\frac{1}{k^{\alpha}}-\frac{1}{(k+1)^{\alpha} })\simeq 
cN.\frac{k^{\alpha-1}}{k^{2\alpha}}=\frac{cN}{k^{\alpha+1}}$$

By lemma \ref{sb}, $c_K \leq \frac{N}{K}$, the uniform average. If
$$c_K \leq  \frac{N}{K}\leq \frac{cN}{k^{\alpha+1}}$$
 the condition \eqref{e1} on  $ f(k) - f(k+1) $ is  guaranteed and we have an exact solution. Hence:
$$K = \Omega(k^{\alpha+1})$$

The new analysis of lemma \ref{ss} gives  a better bound on $K$.
\begin{lemma}
For the Zipf distribution with parameter $\alpha>1$, $K = \Omega(k^{1+1/\alpha})$
guarantees an exact solution.
\end{lemma}

\begin{proof}
By lemma \ref{zb}, $c_K \leq \frac{\Theta(N)}{K^{\alpha}}$ hence if:

$$ c_K \leq  \frac{\Theta(N)}{K^{\alpha}}\leq \frac{\Theta(N)}{k^{\alpha+1}}$$
the condition \eqref{e1}  is guaranteed. Hence
$$K \geq \Theta(k^{1+1/\alpha})$$
\end{proof}

A similar bound is given in \cite{MAE2005}, by arguing that $f(i) < c_K$ for $i >K$. In particular
for $i=K+1$:
$$f(K+1)= \frac{cN}{(K+1)^{\alpha}} \leq  c_K  \leq \frac{cN}{k^{\alpha+1}}$$
which gives $K \geq \Theta(k^{1+1/\alpha})$.\\

For an approximate solution, we take a table $T$ such that:
\begin{equation}
c_K \leq \eps.f(k)    ~~~~ \tag{2}\label{e2}
\end{equation}

\begin{lemma}
For the Zipf distribution with parameter $\alpha >1$, 
$K = \Omega(k. (\frac{1}{\eps})^{1/\alpha} )$
guarantees an approximate solution.

\end{lemma}

\begin{proof}
By lemma \ref{ss}, $c_K \leq \frac{\Theta(N)}{K^{\alpha}}$ hence if:

$$ c_K \leq \frac{\Theta(N)}{K^{\alpha}}\leq \eps.f(k)= \eps. cN. \frac{c}{k^{\alpha}}$$
the condition \eqref{e2}  is guaranteed. Hence:
$$K \geq \Omega(k. (\frac{1}{\eps})^{1/\alpha} )$$

\end{proof}

\subsection{ Proof of the Spacesaving Lemma}\label{l2}

Let $K$ denote the size of the table and $c_K$ be as defined in the SpaceSaving algorithm
Recall that Lemma~\ref{lemma:diffcounts} states that 
given $r\leq K$, let $\tilde{e}$ denote the element with the $r$th largest count in the table, and $e'$ denote the element with the $r$th largest frequency in the stream. We can then bound 
$|\ccount(e')-\ccount(\tilde{e})|$ by $c_K.$

 \begin{figure}[h]
\includegraphics[width=130mm]{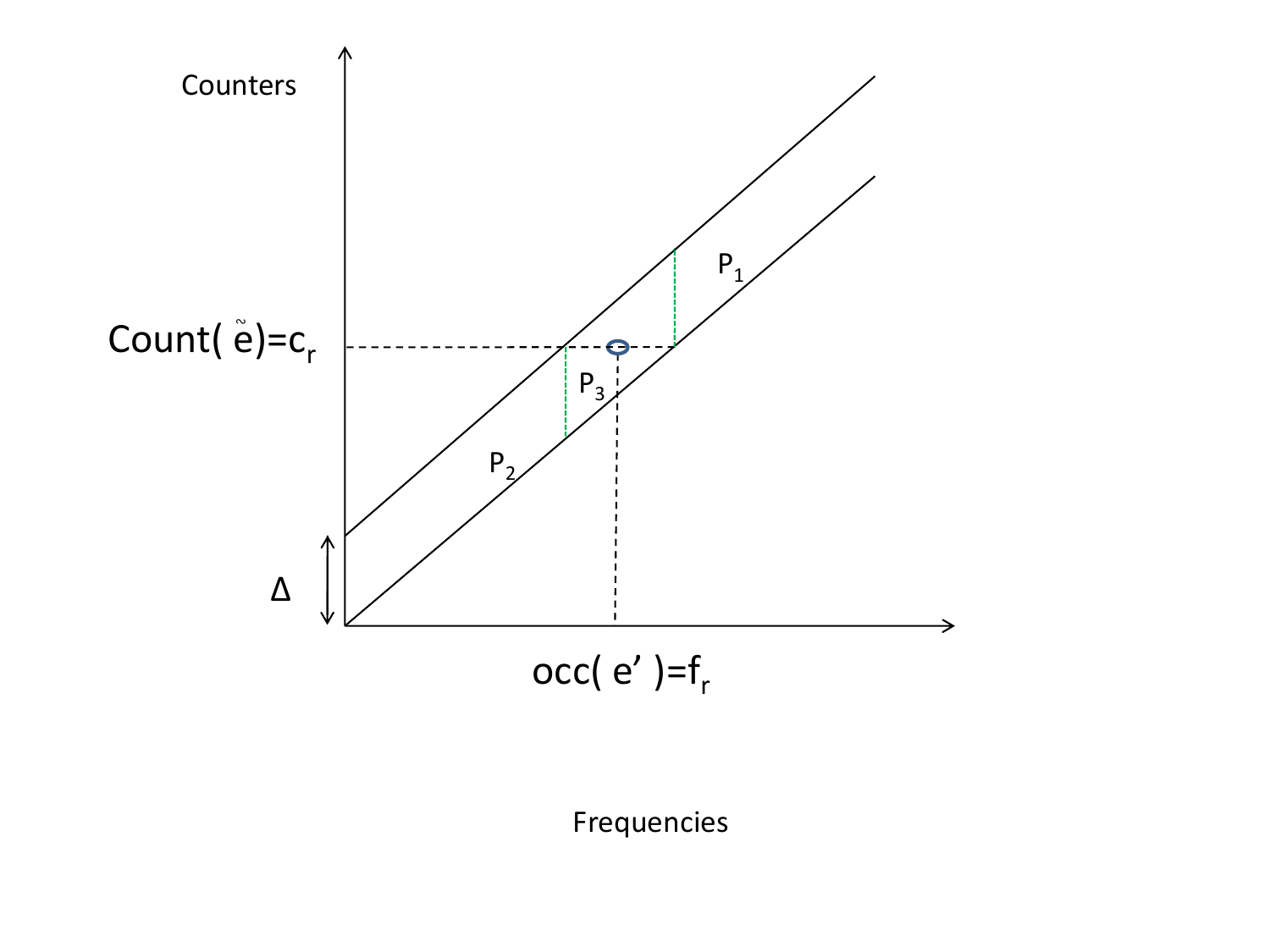}
\caption{Counters and Frequencies for a stream $s$.  The error $\Delta=c_{K}$ and $|\occ(e')-\ccount(\tilde{e})|<c_{K}.$}
\label{fig:3parts.pdf}
\end{figure}
\begin{proof}
We will argue that
\begin{equation}\label{eq:withfigure}
|\ccount(e')-\ccount(\tilde{e})|\leq  c_{K}.
\end{equation}
To that end, we refer the reader to Figure~\ref{fig:3parts.pdf}. By Property~\ref{prop:3}, for any element $e$ of  $s$ we have $  \occ(e)\leq \ccount(e)\leq \occ (e)+c_{K}$, so when we plot the points $(\occ(e),\ccount(e))$ for the elements occuring in stream $s$, all points are inside the strip of equation $x\leq y\leq x+c_{K}$. Consider the point $(\occ(e'),\ccount(\tilde{e}))$. We partition the strip into three parts (see Figure~\ref{fig:3parts.pdf}):
\begin{enumerate}
\item $P_1$ consisting of the points $(x,y)$ such that $x>\ccount(\tilde{e})$. Since $\tilde{e}$ has rank $r$ according to $\ccount$, there are at most $r-1$ points in $P_1$.
\item $P_2$ consisting of the points $(x,y)$ such that $x<\ccount(\tilde{e})-c_{K}$. Since $\tilde{e}$ has rank $r$ according to $\ccount$,  there are  fewer than 
$n -r $ where $n$ is the number of elements in the stream $s$.
%$\#(\hbox{elements})-r$ points in $P_2$.
\item $P_3$ consisting of the rest. All points of $P_1$ have $\occ$ value larger than all points of $P_3$, and 
all points of $P_2$ have $\occ$ value smaller than all points of  $P_3$.
\end{enumerate}
Recall that $e'$ has rank $r$ according to $\occ$. Thus the point $(\occ(e'),\ccount(\tilde{e}))$ cannot be in $P_1$ nor in $P_2$. This implies that
$e'$ is in $P_3$, hence Equation~\ref{eq:withfigure}. 
%%%%%%%%%%%%%%%%%%%%%%%%%%%%%%%%%%%%%

\end{proof}
 The Top-$k$ problem, the  building block used by the Tester
 generalizes for the {\em insertion and  $\alpha$-bounded deletion } model and for the {\em sliding window} model, an {\em insertion and window deletion }  model which is not a bounded deletion model.
 
 % in section \ref{ssavingsw} of the appendix \ref{amo}. In both cases, we have a solution to the Top-$k$ problem, the  building block used by the Tester.

\subsection{ Relative error for  $(\gamma_1,\gamma_2)$-decreasing  functions}\label{tkre}

%\subsection{Properties of $(\gamma_1,\gamma_2)$-decreasing  functions}

%Let $F^{res(k)}=\sum_{k+1\leq i\leq n} f(i)$ be the tail of the frequency  distribution.  
\begin{lemma}\label{tail}
If $f$ is $(\gamma_1,\gamma_2)$-decreasing then for all $k$:
$$\frac{\eps }{ k }.F^{res(k)}\leq \eps .f(k).\frac{(\gamma_1-1)}{1-\gamma_1/\gamma_2}$$

\end{lemma}
\begin{proof}
If $f$ is $(\gamma_1,\gamma_2)$-decreasing then for $j\geq 0$:
 $$\sum_{i> \gamma_1^j.k}^{i=\gamma_1^{j+1}.k} f(i) \leq 
 \frac{f(k)\cdot(\gamma_1^{j+1}.k -\gamma_1^j.k)}{\gamma_2^j}$$
Hence:
$$F^{res(k)}=\sum_{k+1\leq i\leq n} f(i) \leq k.f(k).(\gamma_1-1).\sum_{j\geq 0}\frac{\gamma_1^j}{\gamma_2^j}=k.f(k).(\gamma_1-1).\frac{1}{1-\gamma_1/\gamma_2}$$
$$=k.f(k).\frac{(\gamma_1-1)}{1-\gamma_1/\gamma_2}$$
$$\frac{\eps }{ k }.F^{res(k)}\leq \eps .f(k).\frac{(\gamma_1-1)}{1-\gamma_1/\gamma_2}$$
\end{proof}

We can use this bound  to obtain a relative error on the estimation of the Top frequencies.\\

\begin{lemma}\label{decreas}
If $f$ is $(\gamma_1,\gamma_2)$-decreasing then for all $i\leq k$:
$$|f(i)-c_i | \leq  \eps.f(i).\frac{(\gamma_1-1)}{1-\gamma_1/\gamma_2}$$
\end{lemma}
\begin{proof}
 We first apply  Lemma \ref{ss}:
 $$c_K \leq \min_{u<K/2} \frac{F^{res(u)}}{K-2u}$$
 Hence for $u=k$ and $K=k \cdot(2+ 1/\eps)$:
 $$c_K \leq \frac{F^{res(k)}}{K-2k}\leq \frac{k}{\eps}. F^{res(k)}  $$
We then apply 
 Lemma \ref{tail} and obtain:
 $$|f(i)-c_i | \leq c_K \leq \frac{k}{\eps}. F^{res(k)} \leq \eps.f(k).\frac{(\gamma_1-1)}{1-\gamma_1/\gamma_2}$$
 By montonicity of $f$:
 $$|f(i)-c_i | \leq \eps.f(i).\frac{(\gamma_1-1)}{1-\gamma_1/\gamma_2}$$
 \end{proof}
 
If we take the strong bound from \cite{B10} and combine it with Lemma \ref{tail}, we also obtain the relative error bound, for the top-$k$ frequencies $f(i)$ where $i\leq k$.

\section{Streaming Testers to compare frequency distributions}\label{mt}

A stream $s$ of $N$ elements of a universe $\{ e_1,e_2,\cdots ,e_n\}$ of size $n$  determines an integer  frequency function $g$ whose domain
is $ \{ 1,...n\}$, such that $g(i)$ is the number of occurences of the $i$th most frequent element in the stream. The frequency function $f$  is given and we now show that
Algorithm~\ref{alg:Tester} 
is a streaming Tester to compare $f$ and $g$.

\subsection{Analysis of the space used by Algorithm~\ref{alg:Tester}  }

 \begin{lemma}
 Algorithm~\ref{alg:Tester} uses $O((\log n)^2 \cdot \log\log n)$ space.
 \end{lemma}
\begin{proof}
We  run the Spacesaving  algorithm for each substream $s_i$ for $i\leq O(\log n)$, with a table of size
$K$ where:
%Approximate counting. The size $K$ of the Spacesaving table $T$,   can be bounded. Notice that $p_i=\gamma^{\log (1/\eps_2)}.z_i$ and 

% Hence $$K_i=\frac{p_i}{a_i}=\gamma^{\log (1/\eps_2)}/\eps_1^2=\gamma^{\log (1/\eps_2)}.\frac{z_i}{a_i}$$

$$K = 2\cdot\frac{z_i}{ a_i } \cdot \frac{(\gamma_1-1)}{1-\gamma_1/\gamma_2}\cdot (1+\eps_1^2)\cdot \frac{\log n}{\eps_2 \delta} $$
As $a_i=\lceil \eps_1^2.z_i/ \log\log n)\rceil$.
$$K =\frac{2 }{ \eps_2 . \eps_1^2} \cdot \frac{(\gamma_1-1)}{1-\gamma_1/\gamma_2}\cdot \frac{\log n .\log\log n }{\delta}\cdot(1+\eps_1^2)\leq O(\log n \cdot \log\log n) $$

The total space is then $O((\log n)^2 \cdot \log\log n)$.
 \end{proof}
 
 Notice that  there exists  a threshold  $t_0$ such that for $z_i > t_0$, the expected number of distinct items in $s_i$ is less than $K$. In this case, the Spacesaving is equivalent to exact counting and there is no error.
 As
  $f$ is $(\gamma_1,\gamma_2)$-decreasing, $f(\gamma_1. t) \leq f(t)/\gamma_2$.   Hence 

$$f(\gamma^{\alpha}. t) < f(t)/\gamma_2^{\alpha}=\eps_2.f(t)$$ 
For $\alpha = \log (1/\eps_2)/\log \gamma_2 $ and
 $n=\gamma_1^{\alpha}.t_0$, we find a threshold  $t_0 =n/ \gamma^{\log (1/\eps_2)/\log \gamma_2}$ such that for
$z_i > t_0$ the expected number of elements in the stream $s_i$ is 
$$n/a_i =  n . \log\log n / \eps_1^2.z_i  \leq   n . \log\log n / \eps_1^2.t_0 \leq \gamma^{\log (1/\eps_2)/\log \gamma_2} .  \log\log n / \eps_1^2 \leq K$$
%$K = 2\cdot\frac{z_i}{ a_i } \cdot \frac{(\gamma_1-1)}{1-\gamma_1/\gamma_2}\cdot (1+\eps_1^2)\cdot \frac{\log n}{\eps_2 \delta}$.
The expected number ot items is smaller than the size $K$ of the SpaceSaving Table, and we do 
   an exact counting.

\subsection{Analysis of the error probability of Algorithm~\ref{alg:Tester} }\label{ssaving}

%\newcommand{\ccount}{\hbox{count}}
%\newcommand{\occ}{\hbox{occ}}
%\newcommand{\rankocc}{\hbox{rank}_{\occ}}
%\newcommand{\rankcount}{\hbox{rank}_{\ccount}}
%\newcommand{\rank}{\hbox{rank}}

%%%%%%%Delete

%The intermediate Lemmas are in appendix \ref{c1}.
\begin{notation}  Let  $\tilde{e_i}$ be the element whose counter value is $c_{i,r}$, i.e. $\ccount(\tilde{e_i})=c_{i,r}$ and $e'_i$ the element whose rank is $r=\lceil  z_i/a_i \rceil$ in the stream $s_i$, for the frequency function $g_i$, i.e.  $\occ(e'_i)=g_i(r)$ or $rank_{s_i}(e'_i)=r$. The functions
$\occ, \ccount, rank$ are from $U$ to $N$. We assume that tie-breaking rules are consistent over $s$ and the substreams $s_i$: $U=\{ e^1,e^2,\cdots ,e^n\}$ and if two elements $e^j$ and $e^k$, with $j<k$, have the same number of occurrences, then  $rank_s(e^j)<rank_s(e^k)$ and $rank_{s_i}(e^j)<rank_{s_i}(e^k)$ for all substreams. 
\end{notation} 

%If $g$ is $(\gamma_1,\gamma_2)$-decreasing, we can assume a {\em relative error} in the Spacesaving algorithm. 

We now state the  probabilistic Lemma \ref{lemma:elementary1}, which analyzes the sampling that is used to create the substream $s_i$ and relates the rank of $e'_i$ to $z_i$. This  depends on the sampling process alone and not on the Spacesaving algorithm analysis. The main Lemmas \ref{re1}, \ref{re} guarantee an error bound on Spacesaving on each $s_i$ with high probability. 
%All the proofs are in the appendix \ref{kl}.

%%%%%%%%%%%%%%%%%%%%%%%%%%%%%%%%%%%%%%%%%%%%%%%%%%%%%%%%%
\begin{lemma}\label{lemma:elementary1}
Recall that each element is kept in substream $s_i$ with probability $1/a_i$ and that 
$e'_i$ denotes the element with rank $\lceil  z_i/a_i \rceil$ in substream $s_i$.
% (when sorted in non-increasing order of number of occurences): $rank_{s_i}(e'_i)=z_i/a_i$. 
Then, 
 the rank of $e'_i$ in stream $s$  (when sorted in non-increasing order of number of occurences) satisfies
$$\Pr(z_i(1-\eps_1^2)\leq rank_s(e'_i)\leq z_i (1+\eps_1^2))\geq 1-4\delta/\log n.$$ 
Moreover, if $f=g$ then $f(z_i)\sim_{\eps_2} \occ(e'_i)$.

\end{lemma}
We recall the following classic Hoeffding probabilistic bound.

%%%%%%%%%%%%%%%%%%%%%%%%%%%%%%%%%%%%%%%%%%%%%%%%%%%%%%%%%
\begin{lemma}\label{lemma:chernoff}
Let $X=\sum_{j=1}^p X_i$ where $X_j=1$ with probability $q_j$ and $X_j=0$ with probability $1-q_j$, and the $X_j$'s are independent. 
Let $\mu=\E(X)$. Then for all $0<\beta <1$ we have $$\Pr(|X-\mu|>\beta \mu)\leq 2e^{-\mu\beta^2/3}.$$
\end{lemma}

{\bf Proof of Lemma ~\ref{lemma:elementary1}}.

\begin{proof}
By definition of $e'_i$, the rank of $\occ(e'_i)$ in the substream $s_i$ equals $r=z_i/a_i$.
We will prove the following:
With probability at least $1-4\delta/\log n$, the following properties hold: 
\begin{enumerate}
\item The number of elements that appear in $s_i$ and have rank less than $z_i(1- \eps_1^2)$ in $s$ is less than $z_i/a_i$
\item The number of elements that appear in $s_i$ and have rank less than $z_i(1+\eps_1^2)$ in $s$ is more than $ z_i/a_i$
\end{enumerate}
This will imply the Lemma.

For the first item, we apply Lemma~\ref{lemma:chernoff} with $X$ denoting the number of elements that appear in $s_i$ and have rank less than $p=z_i(1-\eps_1^2)$ in $s$, so that $X_j=1$ if and only if the element of rank $j\leq z_i(1-\eps_1^2)$ in $s$ appears in $s_i$. We have $\mu=z_i(1-\eps_1^2)/a_i$. We set $\beta=\eps_1^2/(1-\eps_1^2)$. We obtain that the probability that the statement  does not hold is at most $2exp(-\frac{  z_i\eps_1^2 }{3a_i (1-\eps_1^2)})\leq 
2exp(-\frac{  z_i\eps_1^2 }{3a_i (1+\eps_1^2)})$. 

For the second item, we apply Lemma~\ref{lemma:chernoff} with $X$ denoting the number of elements that appear in $s_i$ and have rank less than $p=z_i(1+ \eps_1^2)$ in $s$, so that $X_j=1$ if and only if the element of rank $j\leq z_i(1+ \eps_1^2 )$ in $s$ appears in $s_i$. We have $\mu=z_i(1+ \eps_1^2)/a_i$. We set $\beta=\frac{ \eps_1^2 }{(1+ \eps_1^2)}$. We obtain that the probability that the statement does not hold is at most $2exp(-\frac{ z_i\eps_1^2}{3 a_i (1+ \eps_1^2)})$.

By the union bound, the probability that the two statements do not both hold is bounded by
$4exp(-\frac{ z_i\eps_1^2}{3 a_i (1+ \eps_1^2)})$.
Let $a_i= \eps_1^2 z_i / ({6 \ln ((\ln n)/\delta)})$. Then this probability  is at most $4\delta/\ln n$.

Since $f=g$ and $z_i$ is not close to one of the endpoints of the 
%boxes 
Intervals of $f$, we also have $f(z_i)\sim_{\eps_2} \occ(e'_i)$.
\end{proof}

%%%%%%%%%%%%%%%%%%%%%%%%%%%%%%%%%%%%%%%%%%%%%%%%%%%%%%%%%

Now we turn to the analysis of the SpaceSaving algorithm, with two cases:  $z_i\leq 1/\epsilon_1^3$ and $a_i=1$ and we have one stream, Lemma \ref{re1} and the general case when $z_i > 1/\epsilon_1^3$ and we have many substreams, Lemma  \ref{re} which determines the final size $K$ of the tables $T_i$.\\

%%%%%%%%%%%%%%%%%%%%%%%%%%%%%%%%%%%%%%%%%%%%%%%%%%%%%%%%%
\begin{lemma}\label{re1}
Let $i$ be such that $z_i\leq 1/\epsilon_1^3$. 
Assume that $g$ is  $(\gamma_1,\gamma_2)$-decreasing. Consider Algorithm~\ref{alg:Tester} with
$K = 2\cdot\frac{z_i}{ a_i } \cdot \frac{(\gamma_1-1)}{1-\gamma_1/\gamma_2}\cdot (1+\eps_1^2)\cdot \frac{\log n}{\eps_2 \delta}$. 
% $K \geq  {4 (z_i/a_i)}{  } \cdot \frac{(\gamma_1-1)}{1-\gamma_1/\gamma_2}\cdot \frac{\log n}{\eps_2 \delta}$. 
We have: 
$$ c_{K} \leq \eps_2. g(z_i) .$$
\end{lemma}
%%%%%%%%%%%%%%%%%%%%%%%%%%%%%%%%%%%%%%%%%%%%%%%%%%%%%%%%%

%{\bf Proof of Lemma ~\ref{re1}}.
\begin{proof}
Since $i$ is small, $a_i=1$ and stream $s_i=s$. The frequency function of substream $i$ is $g$. For table $T_i$ of size $K$ used by the algorithm. The number of distinct elements in stream $s$ equals $n$. Then the domain of $g$ is $[1,n]$.
The length of stream $s$ equals $N$.
Let $G(u)=\sum_{j=1}^u g(j)$ denote the cumulative frequency, and 
$G^{res(u)}=\sum_{j=u+1}^{ni} g(j)$.
We apply Lemma~\ref{ss} to table $T_i$, using $u={z_i}$ and noting that $K-2{z_i} > K/2$:

\begin{align}
c_{K} \leq 
\min_{u<K/2} \frac{G^{res(u)}}{K-2u} \leq 
\frac{\sum_{{z_i}+1}^{n} g(x)}{K-2{z_i}}\leq  \frac{2.}{K} \sum_{{z_i}+1}^{n} g(x).    \label{eq:1bis}
\end{align}
As in Lemma~\ref{lemma:elementary1}, let $e'_i$ denote the element of substream such that $rank_{s}(e'_i)=z_i$.
We have:
$$
\sum_{{z_i}+1}^{n} g(x)=\sum_{y=rank_s(e'_i)+1}^n g(y).$$

 Now, since $g$ is $(\gamma_1,\gamma_2)$-decreasing, applying Lemma~\ref{tail} to $g(z_i)$ and rewriting, we have:
 
\begin{align}
\sum_{y=z_i }^n g(y)  \leq \frac{(\gamma_1-1)}{1-\gamma_1/\gamma_2} .z_i  .g(z_i) \label{eq:3bis} 
\end{align}
Combining the inequalities \eqref{eq:1bis}  and \eqref{eq:3bis}  gives:

$$c_{K}\leq  \frac{2}{K}\cdot  \frac{(\gamma_1-1)}{1-\gamma_1/\gamma_2} z_i\cdot g(z_i).
$$

As $K = 2\cdot z_i \cdot \frac{(\gamma_1-1)}{1-\gamma_1/\gamma_2}\cdot (1+\eps_1^2)\cdot \frac{\log n}{\eps_2 \delta}$. 
%As $K \geq {4 z_i}{  } \cdot \frac{(\gamma_1-1)}{1-\gamma_1/\gamma_2}\cdot \frac{\log n}{\eps_2 \delta}$, we have:

$$c_{K}\leq    \frac{\delta}{\log n}  .  \eps_2.g(z_i) \leq \eps_2 g(z_i).$$

\end{proof}

We now consider the general case when there are several substreams $s_i$. We use the Spacesaving algorithm for each substream, maintaining a table $T_i$ of size $K$. Let $c_{i,K}$ be the counter value in position $K$ of the table $T_i$. We want a relative error bound on $c_{i,K}$ with high probability.

%%%%%%%%%%%%%%%%%%%%%%%%%%%%%%%%%%%%%%%%%%%%%%%%%%%%%%%%%
\begin{lemma}\label{re}
Let $i$ be such that $z_i> 1/\epsilon_1^3$. 
Assume that $g$ is   $(\gamma_1,\gamma_2)$-decreasing. Consider the Spacesaving on table $T_i$ for the $i$-th substream  in Algorithm~\ref{alg:Tester} and recall that $K = 2\cdot\frac{z_i}{ a_i } \cdot \frac{(\gamma_1-1)}{1-\gamma_1/\gamma_2}\cdot (1+\eps_1^2)\cdot \frac{\log n}{\eps_2 \delta}$. 
We have: 
$$ \Pr[ c_{i,K} \leq \eps_2. g(z_{i+1}) ] \geq 1 -5\delta/\log n.$$
\end{lemma}
%%%%%%%%%%%%%%%%%%%%%%%%%%%%%%%%%%%%%%%%%%%%%%%%%%%%%%%%%
%\begin{theorem-non}\label{main}
%Algorithm $A(s,\eps_1,\eps_2,f)$ is a streaming $4\delta$-Tester if the distribution $f$ is $ (\eps_1,\eps_2)$\hbox{-half-stable}. If $f$ is $(\gamma_1,\gamma_2)$-decreasing, then $A$ uses space $O(\log n)$.
%\end{theorem-non}

%{\bf Proof of Lemma ~\ref{re}}.

\begin{proof}
Let  $g_i$ be the frequency function of substream $i$.  Let $n_i$ denote the number of distinct elements in substream $s_i$. Then the domain of $g_i$ is $[1,n_i]$, and $n_i$ is a random variable with expectation equal to $n/a_i$. Let $N_i$ denote the length of substream $s_i$: we have $N_i =\sum_{x=1}^{x=n_i} g_i(x)$. Let $G_i(u)=\sum_{j=1}^u g_i(j)$ denote the cumulative frequency of the most frequent items, and 
$G_i^{res(u)}=\sum_{j=u+1}^{n_i} g_i(j)$.
Let  $$\widehat{z}=z_{i+2}/a_i$$
 We apply Lemma~\ref{ss} to table $T_i$, using $u=\widehat{z}$ and noting that $K-2\widehat{z} > K/2$:

\begin{align}
c_{i,K} \leq 
\min_{u<K/2} \frac{G_i^{res(u)}}{K-2u} \leq 
\frac{\sum_{\widehat{z}+1}^{n_i} g_i(x)}{K-2\widehat{z}}\leq  \frac{2.}{K} \sum_{\widehat{z}+1}^{n_i} g_i(x).    \label{eq:1}
\end{align}
As in Lemma~\ref{lemma:elementary1}, let $e'$ denote the element of substream $s_i$ such that $rank_{s_i}(e')=\widehat{z}$. By definition of $g_i$ and $e'$, we have:
$$
\sum_{\widehat{z}+1}^{n_i} g_i(x)=\sum_{y=rank_s(e')+1}^n g(y) {\bf 1}(\hbox{the element of $s$ with rank $y$ is in $s_i$}).
$$

Let $A$ denote the following event:
$$rank_s(e')\geq z_{i+1}
$$
Assume that $A$ holds. Then
$$
\sum_{\widehat{z}+1}^{n_i} g_i(x)\leq \sum_{y=z_{i+1}+1}^n g(y) {\bf 1}(\hbox{the element of $s$ with rank $y$ is in $s_i$})
%\sum_{\widehat{z_i}+1}^{n_i} g_i(x)\leq \sum_{z_i(1-\eps_1^2)}^n g(y){\bf 1}(y\in s_i).
$$
Observe that the value of the right-hand side is determined by which elements of $s$ are put in $s_i$, among the ones with $rank_s$ greater than $z_{i+1}$. Also observe that event $A$ is determined by how many elements of $s$ are put in $s_i$, among the ones with $rank_s$ smaller than or equal to $z_{i+1}$. Thus 
the expression in the right-hand side is independent of event $A$, and we can write: 

\begin{eqnarray*}
\E[\sum_{\widehat{z}<x\leq n_i}g_i(x) |A ]&\leq &
	\E[\sum_{y=z_{i+1}+1}^n g(y) {\bf 1}(\hbox{the element of $s$ with rank $y$ is in $s_i$}) |A]\\
	&= &
	\E[\sum_{y=z_{i+1}+1}^n g(y) {\bf 1}(\hbox{the element of $s$ with rank $y$ is in $s_i$}) ]\\
	&=&
	\frac{1}{a_i } \sum_{y=z_{i+1}+1}^n g(y) 
\end{eqnarray*}

 Now, since $g$ is $(\gamma_1,\gamma_2)$-decreasing, applying Lemma~\ref{tail} to $g$ and $k=z_{i+1}$ and rewriting, we have:
 \begin{align}
\sum_{y=z_{i+1} +1}^n g(y)  \leq \frac{(\gamma_1-1)}{1-\gamma_1/\gamma_2} .z_{i+1} .g(z_{i+1}) \label{eq:3} 
\end{align}
%Since $z_i$ is not close to an %Box Interval endpoint of $g$, by Lemma~\ref{lemma:stepcompatible} we have $$g(z_i(1-\eps_1^2))\leq g(z_i)(1+\eps_2).$$ 
%Notice that $ z_{i+2}  (1-\eps_1^2) \simeq z_{i+1}$ and that  $g(z_{i+2}  (1-\eps_1^2) ) \leq g(z_{i+1})$ 
Combining the inequalities \eqref{eq:1}  and \eqref{eq:3}  and let $\mu=\E[c_{i,K}|A ]$ :

$$\mu=\E[c_{i,K}|A ]\leq  \frac{2}{K}\cdot \frac{1}{a_i }\cdot  \frac{(\gamma_1-1)}{1-\gamma_1/\gamma_2} z_{i+1} \cdot g(z_{i+1} )
$$
$$\mu=\E[c_{i,K}|A ]\leq  \frac{2}{K}\cdot \frac{z_i}{a_i }\cdot  \frac{(\gamma_1-1)}{1-\gamma_1/\gamma_2} (1+\eps_1^2) \cdot g(z_{i+1} )
$$

Let  $K = 2\cdot\frac{z_i}{ a_i } \cdot \frac{(\gamma_1-1)}{1-\gamma_1/\gamma_2}\cdot (1+\eps_1^2)\cdot \frac{\log n}{\eps_2 \delta}$, we have:

%[Check   $z_{i+2} (1-\eps_1^2) \simeq z_{i+1} \simeq z_{i} (1+\eps_1^2) $. Change $K$. ]

$$\mu=\E[c_{i,K}|A ]\leq    \frac{\delta}{\log n}  .  \eps_2.g(z_{i+1} ) $$

We use Markov's inequality to conclude that, conditioned on event $A$ we have:
$$\Pr[c_{i,K} \leq \frac{\log n}{\delta}\cdot \mu ~\mid~A]\geq 1-\delta/\log n. $$
$$\Pr[c_{i,K} \leq \eps_2. g(z_{i+1} )  ~\mid~A]\geq 1-\delta/\log n. $$

By Lemma~\ref{lemma:elementary1} event $A$ has probability at least $1-4\delta/\log n$. We conclude that 
$$ \Pr[c_{i,K} \leq \eps_2. g(z_{i+1} )  ] = \Pr[A ]  \cdot \Pr[c_{i,K} \leq \eps_2. g(z_{i+1} )  ~\mid~A]$$
$$ \Pr[c_{i,K} \leq \eps_2. g(z_{i+1} )  ] \geq (1 -4\delta/\log n)(1 - \delta/\log n)\geq  1-5\delta/\log n.$$
\end{proof}

%\subsubsection{Proofs of the key Lemmas}\label{kl}

%%%%%%%%%%%%%%%%%%%%%%%%%%%%%%%%%%%%%%%%%%%%%%%%%%%%%%%%%

%We now prove the  probabilistic Lemma \ref{lemma:elementary1}, and the Lemmas \ref{re1}\,ref{re} which guarantee an error bound on Spacesaving on each $s_i$ with high probability.

We can now prove our main Theorem: \\
%Second appearance of the main theorem
\begin{T1}
\label{mt2}
  \thmtext
\end{T1} 
\begin{proof}

As a preliminary, we first prove that for $z_i\leq 1/\eps_1^3$
\begin{equation}\label{eq:g-versus-counter}
|g(z_i)-c_{r} | \leq 2c_{K}.
\end{equation}
Recall that  $\tilde{e}$ denotes the element of $s$ with the $r-$th largest count  $c_{r}$ and that $e'$ denotes the element of $s$ with the $r-$th largest frequency, so that $g(z_i)=\occ(e')$. We write: 
$$|g(z_i)-c_{r}| =  |\occ(e')-\ccount(\tilde{e})| \leq
|\occ(e')-\ccount(e')|+
|\ccount(e')-\ccount(\tilde{e})|.$$
By properties~\ref{prop:2} and~\ref{prop:3} of Spacesaving (see page \pageref{prop:3}) the first term is at most $c_{K}$, and  by Lemma~\ref{lemma:diffcounts}, 
the second term is also at most $c_{K}$, hence Equation~\ref{eq:g-versus-counter}.\\

To prove the Theorem, in the first part of the analysis, we assume that $f=g$ and aim to prove that the algorithm outputs YES with probability $1-O(\delta)$. 

Consider the (easy) case where $z_i\leq 1/\eps_1^3$. Then $a_i=1$, $r=\lceil z_i\rceil$, the substream $s_i$ equals the stream $s$, so there is no randomness. Since $f=g$, we have $|f(z_i)-c_r| = |g(z_i)-c_r|$.
Using Equation~\ref{eq:g-versus-counter}, we deduce that $|f(z_i)-c_r| \leq 2c_K$. 
By~ Lemma~\ref{re1} $c_{i,K} \leq \eps_2 g(z_i)=\eps_2 f(z_i)$. Hence:

$$~~~~~~~~|f(z_i)-c_r| \leq
2\eps_2 f(z_i)$$
By Lemma~\ref{lemma:losing-epsilon}, this implies $c_r\simeq_{2\eps_2/(1-2\eps_2)}f(z_i)$, so $c_r\simeq_{3\eps_2}f(z_i)$. 
% we have$|\ccount(e'_i)-\ccount(\tilde{e_i})|\leq  c_{K}. $
and the coherence test answers YES.

Now we turn to the case where $z_i> 1/\eps_1^3$.
We prove that  $c_{i,r}$ is close to one of $\{f(z_{i-1}), f(z_{i+1})   \}$ with high probability.
Let $e'_i$ the element of rank $r$ in the substream $s_i$ and let $z=\rank_s(e'_i)    $. By 
Lemma~\ref{lemma:elementary1} with h.p.
$$z_{i-1} \leq z  \leq z_{i+1}$$

%*************************** Without Lemma 3

Because $g$ is a $(3\eps_1,\eps_2)$-half-stable function, consider the rectangle $R=
[x, x.(1+3\eps_1)] \times [y, y(1+\eps_2) ]$  which contains $z$,  guaranteed 
by the half-stability hypothesis. Recall that $z_{i+1}=(1+\eps_1 ^2)^2 .z_{i-1}$. Because $ (1+\eps_1 ^2)^2 < (1+3\eps_1)$,  
   at least one of $z_{i-1}$ and $z_{i+1}$ is in $R$.
Assume $R$ contains $z_{i+1}$, then:  
$$~~~~~~~~|f(z_{i+1})-c_{i,r}| = |g(z_{i+1})-c_{i,r}| ~~ as~ f=g~~~~~~~~~~~~~~~~~~~~~~~~$$
By the triangular inequality:
\begin{align}
 ~~~~~~~~|f(z_{i+1})-c_{i,r}| \leq |g(z_{i+1})-\occ(e'_i)|+|\occ(e'_i)-c_{i,r}|
 \label{eq:t1a}
\end{align}
The first term is:
$$|g(z_{i+1})-\occ(e'_i)|=|g(z_{i+1})-g(z)|= (\frac{g(z)}{g(z_{i+1})}-1) \cdot g(z_{i+1}))\leq \eps_2 . g(z_{i+1})$$
as $z$ and $z_{i+1}$ are both in $R$.
%$$ |g(z_{i+1})-\occ(e'_i)| \leq \eps_2 . g(z_{i+1})$$
For the second term:
$$|\occ(e'_i)-c_{i,r}|  \leq c_{i,K}$$
 from properties 2 and 3 of Spacesaving (see page \pageref{prop:3}).
 % and the second bound from Lemma \ref{lemma:diffcounts}. (check ****)\\
By Lemma \ref{re} with high probability:
$$ c_{i,K} \leq \eps_2 g(z_{i+1})$$

We conclude  from (\ref{eq:t1a}) that: 
$$|f(z_{i+1})-c_{i,r}| \leq 2\eps_2 . g(z_{i+1})=2\eps_2 . f(z_{i+1})$$
By Lemma~\ref{lemma:losing-epsilon} this implies
$$c_{i,r}. \simeq_{2\eps_2/(1-2\eps_2)} f(z_{i+1}) $$ 
The  coherence test is passed as $c_{i,r} \simeq_{5\eps_2}f(z_{i+1})$. The case when $J$ contains $z_{i-1}$ is similar.\\

In the second part of the analysis, we assume that $g$ is far from $f$, i.e.  $f \not\sim_{ (3.\eps_1,10\eps_2)} g$, and aim to prove that with probability $1-O(\delta)$ the algorithm outputs NO. 

By  Theorem \ref{main0} with $d=10$, there exists
a separating rectangle $R=[b, b(1+\eps_1)]*[c,c(1+8\eps_2)]$ which separates $f$ from $g$, see figure \ref{fig:sr}.  We want to show that $c_{i,r}$ which is close to $g(z_i)$ is far from $f(z_i)$.

Consider the (easy) case where $[b,b(1+\eps_1)]$ contains some $z_i\leq 1/\eps_1^3$ where we have one stream only. Hence $g(z_i)\not\sim_{8\eps_2}f(z_i)$ and by Lemma~\ref{lemma:losing-epsilon} we have: 

\begin{align}
 ~~~~~~~~
|f(z_i)-g(z_i)|>\frac{8\eps_2}{1+8\eps_2}\max(f(z_i),g(z_i)) \geq 7\eps_2 \cdot \max(f(z_i),g(z_i))
 \label{eq:fg}
 \end{align}
  By Equation~\ref{eq:g-versus-counter} and Lemma~\ref{re1}, 
$$|g(z_i)-c_r| \leq
2c_K\leq 2\eps_2g(z_i) \leq  2\eps_2 \cdot \max(f(z_i),g(z_i)) $$
%Assume $f(z_i) >g(z_i)$ and therefore $f(z_i) >c_r$
 Hence: 
 $$|f(z_i)-c_r|\geq |f(z_i)-g(z_i)|-|g(z_i)-c_r|$$
 $$~~~~~~~~~~~~~~~~~~~~~~~~~~~~~~~~~~ \geq (7\eps_2-2\eps_2) \cdot \max(f(z_i),g(z_i))\geq 5\eps_2 \cdot \max(f(z_i),g(z_i))$$
%If $f(z_i) >g(z_i)$ as $c_r$ is close to $g(z_i)$,  $\max(f(z_i),g(z_i)) =f(z_i) >  \min(f(z_i),c_r)=c_r$. If $f(z_i) <g(z_i)$,  $\max(f(z_i),g(z_i)) =g(z_i) >  \min(f(z_i),c_r)=f(z_i)$. 
Notice that $\max(f(z_i),g(z_i)) \geq \min(f(z_i),c_r)$. Hence: 
 $$|f(z_i)-c_r|\geq 5\eps_2 \cdot \max(f(z_i),g(z_i)) \geq  5\eps_2 \min(f(z_i),c_r)$$
 
% Let $\alpha=\frac{6\eps_2}{1+6\eps_2}-2\eps_2$. By Lemma~\ref{lemma:losing-epsilon},  $c_r\not\simeq_{\alpha/(1-\alpha)} f(z_i)$,   

The coherence test answers NO as $c_r\not\simeq_{5\eps_2} f(z_i)$\\

If $z_i\geq 1/\eps_1^3$, there is some $i$ such that $z_{i-1}, z_i, z_{i+1}$, separated by $(1+\eps_1^2)$,  are all in the range of the separating rectangle $R$ of relative width $1+\eps_1$.  
The inequality (\ref{eq:g-versus-counter}) is not necessarly true for $g(z_i)$.
% but only for one of the $g(z_{i-1})$ or $g(z_i)$ or $g(z_{i+1})$ because $g$ is half-stable.
Consider the two cases in Figure \ref{sep-rect}: in the first case $f$ is below the separating rectangle $R$ and $g$ is above, in the second case $g$ is below the separating rectangle $R$ and $f$ is above.

 \begin{figure}[ht]
 \begin{center}
\includegraphics[height=60mm]{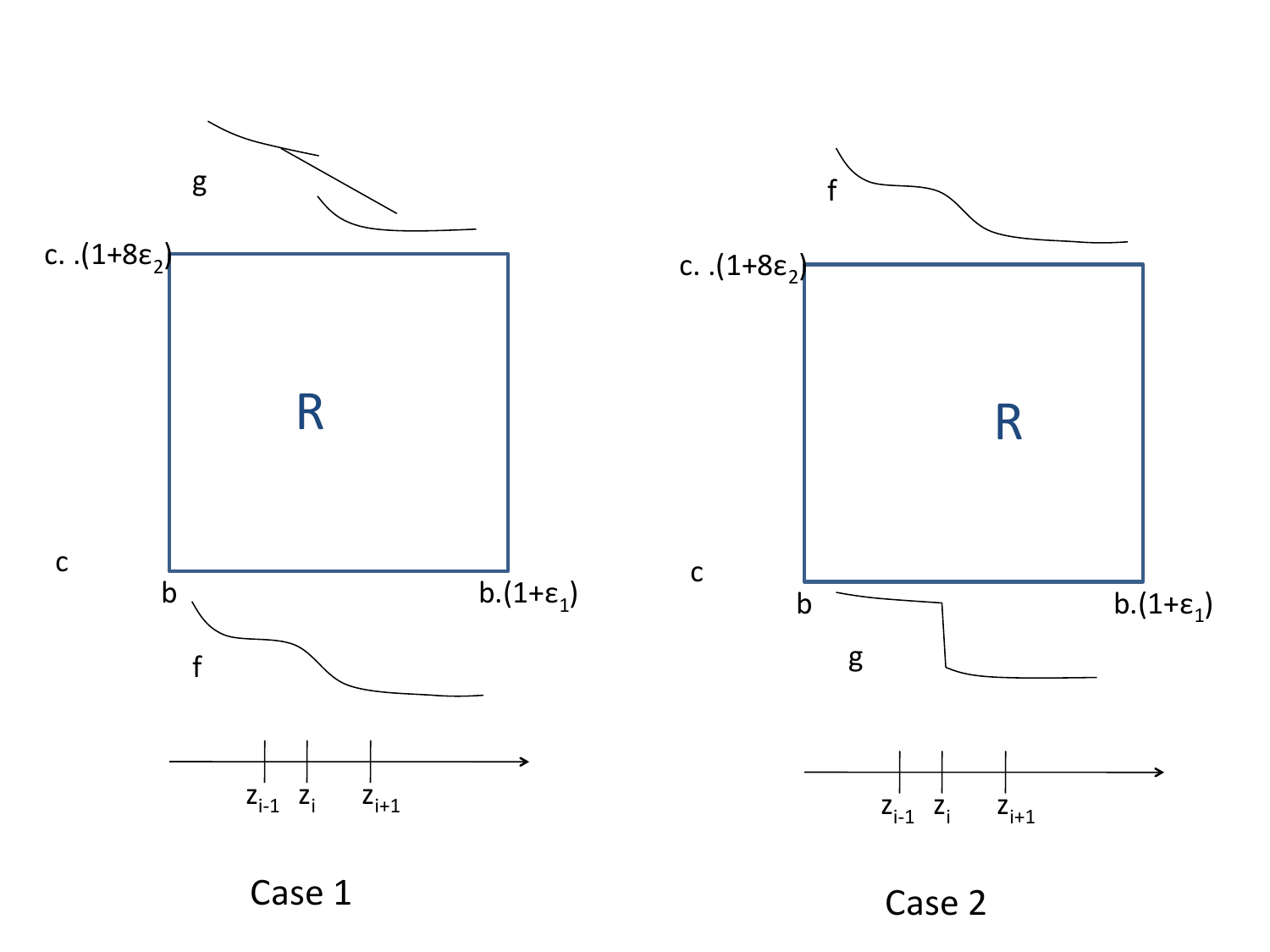}

~~~~~~~~~~~~\caption{~~~~~~~~Separating rectangle $R$ when $f$ and $g$ are far from each other} 
\label{sep-rect}
 \end{center}
\end{figure}

For case 1, $f$ is below $R=[b, b(1+\eps_1)]*[c,c(1+8\eps_2)]$ and:

$$ c_{i,r}= \ccount(\tilde{e_i}) \geq \ccount(e'_i)- c_{i,K}$$
$$ \ccount(e'_i) \geq  \occ(e'_i)- c_{i,K}$$
$$  \occ(e'_i) \geq g(z_{i+1})$$
Then $c_{i,r} \geq g(z_{i+1})-2c_{i,K}$.
By Lemma \ref{re}, with h.p. we have:
$c_{i,K} \leq \eps_2 . g(z_{i+1})$.
Hence $c_{i,r} \geq (1-2\eps_2). g(z_{i+1})$. As $g$ is above $R$:
$g(z_{i+1}) \geq c.(1+8\eps_2)$.
 Hence:
 $$c_{i,r} \geq (1-2\eps_2). g(z_{i+1}) \geq  (1-2\eps_2). c.(1+8\eps_2)$$
Because $f$ is below $R$:
$$f(z_{i+1}) \leq f(z_{i}) \leq f(z_{i-1}) \leq c$$
$$ \frac{c_{i,r}}{f(z_{i+1})} \geq  \frac{c_{i,r}}{f(z_{i})}\geq \frac{c_{i,r}}{f(z_{i-1})} \geq (1-2\eps_2). (1+8\eps_2)\geq (1+5\eps_2)$$
The coherence test answers NO with h.p.\\

For case 2, $f$ is above $R=[b, b(1+\eps_1)]*[c,c(1+8\eps_2)]$ and:
$$ c_{i,r}= \ccount(\tilde{e_i}) \leq \ccount(e'_i)+ c_{i,K}$$
$$ \ccount(e'_i) \leq  \occ(e'_i)+ c_{i,K}$$
$$  \occ(e'_i) \leq c$$
$$  g(z_{i+1})  \leq c$$
Then $c_{i,r} \leq c+ 2c_{i,K}$. By Lemma \ref{re}, with h.p. we have:
$c_{i,K} \leq \eps_2 . g(z_{i+1})$. Hence:
$$ c_{i,r} \leq c.(1+2\eps_2)$$
$$f(z_{i-1}) \geq f(z_{i}) \geq f(z_{i+1}) \geq c.(1+8\eps_2) $$
$$ \frac{f(z_{i-1})}{c_{i,r}} \geq  \frac{f(z_{i})}{c_{i,r}}\geq \frac{f(z_{i+1})}{c_{i,r}} \geq (1+8\eps_2)/(1+2\eps_2)  \geq (1+5\eps_2)$$
The coherence test answers NO with h.p.
\end{proof}

In order to extend the Tester $A_1$ to the Tester $A_2$ which takes two different streams as input, we need to analyse the case when  $f \sim_{(\eps_1, \eps_2)}g$ instead of $f=g$.\\

\begin{corollary}[Tolerant Tester]\label{tt}
If  $f$ and $g$ are $ (3\eps_1,\eps_2)$\hbox{-half-stable} and $(\gamma_1,\gamma_2)$-decreasing 
and $f \sim_{(\eps_1^4 , \eps_2)}g$ then the Tester $A_1$ outputs YES with high probability.
\end{corollary} 
\begin{proof}
Since $f\sim_{(\eps_1^4 , \eps_2)} g$, there exists a $t$ such that  $|t-z_i | \leq \eps_1^4\cdot z_i $ and
 $$(t,f(t))\simeq_{\eps_1^4,\eps2} (z_i,g(z_i))$$
 
Consider the (easy) case where $z_i\leq 1/\eps_1^3$. Then the substream $s_i$ equals the stream $s$, so there is no randomness. 
 Since $z_i\leq 1/\eps_1^3$, then $\eps_1^4\cdot z_i \leq \eps_1$ and $t=z_i$. 
 As  $|g(z_i)-c_r|\leq 2\eps_2 g(z_i)$ and also $f(z_i)\simeq_{\eps2} g(z_i)$, we write:
$$|f(z_i)-c_r|\leq |f(z_i)-g(z_i)|+|g(z_i)-c_r|  \leq \eps_2 \min(f(z_i),g(z_i))+2\eps_2 g(z_i)$$
$$\leq \eps_2 f(z_i) +2\eps_2.(1+\eps_2).f(z_i)= (3\eps_2 +2 \eps_2^2).f(z_i)$$
By Lemma \ref{lemma:losing-epsilon}, we conclude that $|f(z_i)-c_r|\leq 4\eps_2 \min(f(z_{i}),c_{r})$ and
$c_{r}\not\simeq_{5\eps_2} f(z_i)$.

%We have two cases.First, if $g(z_i)\leq \max(f(z_i),c_r)$. Then$$\eps_2 \min(f(z_i),g(z_i))+2\eps_2 g(z_i)\leq 3\eps_2 \max(f(z_i),c_r),$$hence by Lemma~\ref{xxxx} we have $f(z_i)\simeq_{4\eps_2}c_r$ and the test is passed.Second, if $g(z_i)\geq \max (f(z_i),c_r)$. Then$$\eps_2 \min(f(z_i),g(z_i))+2\eps_2 g(z_i)=\eps_2 f(z_i)+2\eps_2 g(z_i)\leq \eps_2f(z_i)+3.5\eps_2f(z_i)$$since $g(z_i)\simeq_{\eps_2}f(z_i)$. Hence $f(z_i)\simeq_{5\eps_2}c_r$ and the test is passed.

Now we turn to the case where $z_i> 1/\eps_1^3$. As $t$ is very close to $z_i$, it is between 
$z_{i-1}$ and $z_{i+1}$. Because $f$ is  $(3\eps_1, \eps_2)$-half-stable, $t$ is in a Box $J$ of relative
 height $(1+\eps_2)$ which includes either $z_{i-1}$ or  $z_{i+1}$.

 In the first case, $t$ and $z_{i+1}$ are in the same box $J$. Then:
 $$|c_{i,r}- f(z_{i+1})|\leq |c_{i,r}- f(t)| + |f(t)- f(z_{i+1})|$$
 $$~~~~~~~~\leq  |c_{i,r}-g(z_i) |+ |g(z_i)- f(t)| + |f(t)- f(z_{i+1})|$$
 Let us bound each part:  $ |c_{i,r}-g(z_i) |\leq 2\eps_2 \cdot g(z_i)\leq (2,1) \eps_2 \cdot c_{i,r}$ by Equation 
 \ref{eq:g-versus-counter} and Lemma \ref{lemma:losing-epsilon}. 
  $ |g(z_i)- f(t)| \leq \eps_2.\min(g(z_i),f(t)) \leq \eps_2.f(t)  \leq \eps_2(1+\eps_2 ).f(z_{i+1}) $
  because of the hypothesis and $t$ and $z_{i+1}$ are in the same box $J$. 
  $|f(t)- f(z_{i+1})| \leq \eps_2.\min(f(t),f(z_{i+1}))\leq \eps_2.f(t)  \leq \eps_2.f(z_{i+1}) $ because $f$ is monotone. 
  Hence:
 $$|c_{i,r}- f(z_{i+1})|\leq (2,1) \eps_2 \cdot c_{i,r}+ \eps_2(1+\eps_2 ).f(z_{i+1})  +\eps_2.f(z_{i+1}) $$
 $$|c_{i,r}- f(z_{i+1})|\leq (4,2 ) \eps_2  \cdot \max(c_{i,r}, f(z_{i+1}))$$
 By Lemma \ref{lemma:losing-epsilon},
 $$|c_{i,r}- f(z_{i+1})|\leq 5 \eps_2 \min(c_{i,r}, f(z_{i+1}))$$
 Hence $c_{i,r}\simeq_{5\eps_2} f(z_{i+1})$.\\

 In the second case, $t$ and $z_{i-1}$ are in the same box $J$. Then:
 $$|c_{i,r}- f(z_{i-1})|\leq |c_{i,r}- f(t)| + |f(t)- f(z_{i-1})|$$
 $$~~~~~~~~\leq  |c_{i,r}-g(z_i) |+ |g(z_i)- f(t)| + |f(t)- f(z_{i-1})|$$
We bound each part: $ |c_{i,r}-g(z_i) | \leq 2\eps_2 \cdot g(z_i)\leq (2,1) \eps_2 \cdot c_{i,r}$ as in the previous case.  $ |g(z_i)- f(t)| \leq \eps_2.\min(g(z_i),f(t)) \leq \eps_2.f(t)  \leq \eps_2.f(z_{i-1}) $ because $f$ is monotone decreasing. $|f(t)- f(z_{i-1})| \leq \eps_2.\min(f(t),f(z_{i-1}))\leq \eps_2.f(t)  \leq \eps_2.(1+\eps_2 ).f(z_{i-1}) $ because $t$ and $z_{i-1}$ are in the same box $J$. Hence:
$$|c_{i,r}- f(z_{i-1})|\leq   (4,2 ) \eps_2  \cdot \max(c_{i,r}, f(z_{i-1}))$$
By Lemma \ref{lemma:losing-epsilon},
 $$|c_{i,r}- f(z_{i-1})|\leq 5 \eps_2 \min(c_{i,r}, f(z_{i-1}))$$
 Hence $c_{i,r}\simeq_{5\eps_2} f(z_{i-1})$.\\
 
The  coherence test is passed and the Tester answers YES.
\end{proof}

%The following Lemma is a stand-alone \marginpar{\tiny C: this lemma is probably already known by the SpaceSaving people} property of the SpaceSaving algorithm. \footnote{In the proof of the main Theorem, it will be useful when applied to stream $\sigma=s_i$ and to elements $x'=e'_i$ and $\tilde{x}= \tilde{e_i}$.}

\subsection{A corrector $\bar{g}$ for the distribution $g$ of the stream}\label{corrlemma}

Let  $\widehat{g} $  the step function defined by $\widehat{g}(z_i)=c_{i,r} $ and  $\widehat{g}(t) =\widehat{g}(z_i)$ for $z_i \leq t < z_{i+1}$ in Algorithm $A_1$.  Because of errors in the counter value $c_{i,r} $, the function may not be monotone. We therefore  introduce its monotone version $ \bar{g}$.
$$ \bar{g}(t)=    Min _{y\leq t }\{ \widehat{g}(y) \}   $$

Let us first prove that $\widehat{g}$ and $\bar{g} $ are very close to $g$, depending on the required size $K$ of the SpaceSaving tables.

\begin{lemma} \label{gbara}
If $g$ is $(\eps_1,\eps_2)$ half-stable and  $(\gamma_1,\gamma_2 )$-decreasing, then Algiorithm $A_1$ with $K= \frac{2 }{ \eps_2 . \eps_1^2} \cdot \frac{(\gamma_1-1)}{1-\gamma_1/\gamma_2}\cdot \frac{\log n .\log\log n }{\delta}\cdot(1+\eps_1^2)$ defines $\widehat{g} $  and $ \bar{g}$ such that  $\widehat{g}\sim_{(\eps_1,2\eps_2)}g$ and
  $\bar{g} \sim_{(5\eps_1,2\eps_2)} g$ with high probability. 
\end{lemma}

\begin{proof}

Let us prove the first part of the Lemma. 
Let $z_i \leq t < z_{i+1}$. Then $(t,\widehat{g}(t))=(t,\widehat{g}(z_i))$.
By  the inequality (\ref{eq:g-versus-counter}) in Theorem \ref{mt2}
$$|g(z_i)-c_{i,r} | \leq 2c_{i,K}$$
and by Lemmas~\ref{re} and~\ref{re1}, with h.p.
$$c_{i,K} \leq \eps_2. g(z_{i})$$
Hence $|g(z_i)-\widehat{g}(z_i))| \leq 2 \cdot c_{i,K} \leq 2 \cdot \eps_2. g(z_{i}) $ and
$$(t,\widehat{g}(t)) \simeq_{(\eps_1,2\eps_2)} (z_{i},g(z_{i})).$$

For the second part of the Lemma, let $(t, \bar{g}(t))$ a point of $\bar{g}$ and let us show that is is close to a point
 of $g$. By definition of $\bar{g}$ there is a $u\leq t $ such that $\bar{g}(t))=\widehat{g}(u)$ and $u=z_i$ for some $i$. By the the first part of the Lemma,  there exists $u' $ such that $(u,\widehat{g}(u)) \simeq_{(\eps_1,2\eps_2)}( u', g(u'))$ and  there exists $t'$ such that $(t,\widehat{g}(t))\simeq_{(\eps_1,2\eps_2)} (t', g(t'))$. Let us consider two cases depending on whether $u$ is close to $t$ or not.

Case 1:  $t \simeq_{3\eps_1}  u$. We also have $u \simeq_{\eps_1} u'$, therefore $t \simeq_{5\eps_1} u'$. Hence 
$(t, \bar{g}(t)) \simeq_{5\eps_1,2\eps_2}  (u',g(u'))$.

Case 2:  $t \not \simeq_{3\eps_1} u$. Observe that since $u'\simeq_{\eps_1} u$ and $t'\simeq_{\eps_1}t$ and $u<t$, we must have $u'<t'$.  
We can then write:  $\bar{g}(t)=\widehat{g}(u) \simeq_{2\eps_2} g(u') \geq g(t')$ because $g$ is decreasing then: 
$$\bar{g}(t) \geq   \frac {1}{1+2\eps_2} \cdot g(t').$$
Now $\bar{g}(t) \leq \widehat{g}(t) \simeq_{2\eps_2}  g(t')$, so: 
$$\bar{g}(t) \leq   (1+2\eps_2) \cdot g(t')$$
Hence $\bar{g}(t) \simeq_{2\eps_2}  g(t')$ and $(t,\bar{g}(t))\simeq_{(\eps_1,2\eps_2)} (t',g(t'))$.
In both cases $\bar{g} \sim_{(5\eps_1,2\eps_2)} g$. 
\end{proof}

 Let us study the conditions on $g$ such that $ \bar{g}$ is both   $(\gamma_1,\gamma_2 )$-decreasing and $(\eps_1,\eps_2)$ half-stable.
%\begin{definition}$(\gamma_1,\gamma_2)$-decreasing) Let $\gamma_2>\gamma_1 >1$. 
%A non-increasing function $f$ with domain $\{ 1,\cdots ,n\}$is $(\gamma_1,\gamma_2 )$-decreasing if for all $t$ such that $1 \leq \gamma_1. t \leq n$:$$f(\lceil \gamma_1.t \rceil )\leq f(t)/\gamma_2$$\label{def-decreasing}\end{definition}
Let $\widehat{g}$ be a $(\gamma_1,\gamma_2 )$-decreasing function then $ \bar{g}$, its monotone
decreasing approximation, is also  $(\gamma_1,\gamma_2 )$-decreasing.
\begin{lemma} \label{gbar}
If $\widehat{g}$ is a $(\gamma_1,\gamma_2 )$-decreasing function, its monotone approximation  $ \bar{g}$ is also $(\gamma_1,\gamma_2 )$-decreasing.
\end{lemma}

\begin{proof}
$$ \bar{g}(\gamma_1 t)=    Min _{y\leq \gamma_1 t }\{ \widehat{g}(y) \}   $$

$$~~~~ ~~~~~~~~~ ~~~~ \leq  Min _{~y\leq \gamma_1 t ~}\{ \widehat{g}(y/\gamma_1) /\gamma_2\}  $$
$$~~~~ ~~~~~~ ~~~ ~~~ ~~\leq  \frac{1}{\gamma_2} ~ ~Min _{~y\leq \gamma_1 t ~}\{ \widehat{g}(y/\gamma_1) \}  $$
$$~~~~ ~~~~~~ ~~~  \leq  \frac{1}{\gamma_2} ~ ~Min _{~z\leq \ t ~}\{ \widehat{g}(z) \}  $$
$$ \leq  \frac{1}{\gamma_2} ~ ~\bar{g}(t)   $$

\end{proof}
\begin{lemma} \label{gbarb}
If $g$ is a $(\gamma_1/(1+\eps_1),(1+\eps_2)^2\gamma_2 )$-decreasing function, its monotone approximation  $ \bar{g}$ is  $(\gamma_1,\gamma_2 )$-decreasing.
\end{lemma}
\begin{proof}
By definition $ g(\gamma_1  t /(1+\eps_1) ) \leq  g(t) /(1+\eps_2)^2\gamma_2  $.

Recall that  by lemma \ref{re1}, with high probability:
$$   g(z_i) /(1+\eps_2) \leq \widehat{g}(z_i) \leq  g(z_i) (1+\eps_2) $$ 
For  $z_i \leq t < z_{i+1}$  then $\widehat{g}(t)=\widehat{g}(z_i) \geq g(z_i) /(1+\eps_2) \geq g(t)/(1+\eps_2)$ with h.p.
We also have with h.p.: 
$$\widehat{g}(t)=\widehat{g}(z_i) \leq g(z_i) (1+\eps_2)  \leq g(t / (1+\eps_1) )(1+\eps_2)$$
because $t / (1+\eps_1) < z_i$ and $g$ is monotone. Hence:
$$   g(t) /(1+\eps_2) \leq \widehat{g}(t) \leq  g(t/ (1+\eps_1) )  (1+\eps_2)  ~~~{\rm and} $$ 
$$  \widehat{g}(\gamma_1 t)  \leq (1+\eps_2)   g(\gamma_1 t / (1+\eps_1) )  \leq  (1+\eps_2)  g(t) /(1+\eps_2)^2 \gamma_2 =g(t) /(1+\eps_2)\gamma_2 \leq  \widehat{g}(t) /\gamma_2 $$
because  $g$ is a $((1+\eps_1)\gamma_1,(1+\eps_2)^2\gamma_2 )$-decreasing function.
Hence and $\widehat{g}$ is a $(\gamma_1,\gamma_2 )$-decreasing function.
By lemma \ref{gbar}, $ \bar{g}$ is  also $(\gamma_1,\gamma_2 )$-decreasing
\end{proof}

\begin{lemma} \label{gbarc}
If $g$ is a $(\eps_1,\eps_2)$ half-stable  function, its monotone approximation  $ \bar{g}$ is  $(2\eps_1,2\eps_2)$  half-stable.
\end{lemma}
\begin{proof}
If $g$ is a $(\eps_1,\eps_2)$ half-stable, let $R$ be the $(\eps_1,\eps_2)$ rectangle associated with $(z_i, g(z_i))$ which includes $g$. Let $z_i \geq t  \geq z_{i+1}$ and the point $(t, \bar{g}(t))$. Increase $R$ to include the point $(t, \bar{g})$, defining $R'$. The span of $R'$ includes $\bar{g}$ which is then $ \bar{g}$ is  $(2\eps_1,2\eps_2)$  half-stable.
\end{proof}

\subsection{Comparison between frequency distributions}\label{comp}

In this section we generalize the Tester Algorithm $A_1$ and introduce two generalizations:  when we have two distinct streams and  when we consider two frequency distributions corresponding to two marginal distributions of the same stream.

\subsubsection{Comparison between the frequency distributions of two streams}\label{2streams}

Suppose we have two independent streams  $s_1$ and $s_2$ and we want to compare  their frequency distributions $g_1$ and $g_2$. 
As with Algorithm $A_1$, we sample  both streams $s_1$ and $s_2$  and for each substream $s_i^1$ and $s_i^2$ we have two tables $T_i^1$  and $T_i^2$ sorted by counter values. The first  table $T_i^1$ applies to the elements of $s_i^1$ 
 and let $c_{i,r}^1$ be the counter value in position $r$ of the table $T_i^1$. Similarly, the table $T_i^2$ applies to the elements of $s_i^2$
 and let
 $c_{i,r}^2$  be the counter value in position $r$.
 
 We will use $\bar{g}_1$  built from the $\{z_i, c_{i,r}^1)$ introduced in section  \ref{corrlemma} as a function $f$ in the Tester $A_1$. As it approximates $g_1$, we will require $\bar{g}_1 \sim_{(\eps_1^4 , \eps_2)} g_1$, i.e. a very good approximation of $g_1$ on the $x$-axis and consequently the size  $K_1$ of the tables $T_i^1$ will be much bigger than the size $K_2$ of the tables $T_i^2$. Let:
%\begin{footnotesize}

$$K_1=\frac{4 }{ \eps_2 . (\eps_1^4/5)^2} \cdot \frac{(\gamma_1-1)}{1-\gamma_1/\gamma_2}\cdot \frac{\log n .\log\log n }{\delta}\cdot(1+(\eps_1^4/5)^2)$$

$$K_2=\frac{2 }{ \eps_2 . \eps_1^2} \cdot \frac{(\gamma_1-1)}{1-\gamma_1/\gamma_2}\cdot \frac{\log n .\log\log n }{\delta}\cdot(1+\eps_1^2)$$

%\end{footnotesize}

\begin{T3}
  \thmtextc
\end{T3}

\begin{proof}
We use the corrector $\bar{g}_1$  built from the $\{z_i, c_{i,r}^1)$ of  $g_1$ introduced in section  \ref{corrlemma}. We then follow the analysis of Theorem \ref{mt2}, using $\bar{g}_1$  for the function 
$f$ and $g_2$ for $g$. If the hypothesis of Lemma \ref{gbara} is satisfied, then $\bar{g}_1 \sim_{(5\eps_1,2\eps_2)} g_1$ and if we set $\eps'_1= \eps_1^4/5$ and $\eps'_2= \eps_2/2$, then 
$\bar{g}_1 \sim_{(\eps_1^4 , \eps_2)} g_1$. We also  replace  $\eps_1, \eps_2$ in the value of $K$ used by the Tester $A_1$ by 
$\eps'_1, \eps'_2$  and  obtain $K_1$ defined above, whereas $K_2=K$.

We first have to verify the hypothesis of Lemmas \ref{gbara}, \ref{gbar}, \ref{gbarb}, \ref{gbarc}, Theorem \ref{main} and use the conclusions of  Theorem \ref{main}  and Corollary \ref{tt}.\\

By hypothesis  $g_1$ is   $(3\eps_1/2, \eps_2/4)$ half-stable. By  Lemma \ref{gbarc} the function $\bar{g}_1 $ is $(3\eps_1, \eps_2/2)$ half-stable, hence also 
$(\eps_1^4/5, \eps_2/2)$ half-stable because\footnote{If a frequency function $g$  is $(\eps_1, \eps_2)$  half-stable it is also $(\eps'_1, \eps'_2)$ for $\eps'_1\leq \eps_1$ and $\eps'_2\geq \eps_2$  because if the curve $g$ is an Box $J$, it is also in a box $J'$ smaller on the $x$-axis but larger on the $y$-axis.}  $\eps_1^4/5 < 3\eps_1$. 
By hypothesis, 
 $g_1$ is a $(\gamma_1/(1+\eps_1),(1+\eps_2)^2\gamma_2 )$-decreasing function,  hence by Lemma \ref{gbarb} the function  $\bar{g}_1$ is  $(\gamma_1,\gamma_2 )$-decreasing. We can therefore apply 
Lemma \ref{gbara} with $\eps'_1= \eps_1^4/5$ and $\eps'_2= \eps_2/2$, and conclude that
 $\bar{g}_1 \sim_{(\eps_1^4 , \eps_2)} g_1$.  Because $\bar{g}_1$ and $g_2$ satisfies the hypothesis of the Theorem 
\ref{main} we  use the same conclusions and in particular     Corollary \ref{tt}.\\

For the positive case, assume that $g_1=g_2$. Because $\bar{g}_1 \sim_{(\eps_1^4 , \eps_2)} g_1$, we infer that
$\bar{g}_1 \sim_{(\eps_1^4 , \eps_2)} g_2$. By the  Corollary \ref{tt}, the Tester $A_2$ accepts with high probability.\\

For the negative case, assume $g_1 \not\sim_{(4\eps_1, 12\eps_2)} g_2$.  We then want to show that as in Theorem \ref{main}
$\bar{g}_1 \not\sim_{(3\eps_1, 10\eps_2)} g_2$. By definition, 

$~~~~~~~~~~\exists t ~~\forall u ~~(t, g_1(t)) \not\simeq_{(4\eps_1, 12\eps_2)} (u,g_2(u))$

Therefore 
%$u \notin  \left[ \frac{t}{1 + 4\eps_1},\; t \cdot (1 + 4\eps_1) \right]$
$u \notin  [ t /(1+4\eps_1), t.(1+4\eps_1)]$ 
or $g_2(u) \notin [g_1(t)/(1 +12\eps_2), g_1(t) \cdot  (1+12\eps_2)]$.
%$g_2(u) \notin \left[  \frac{g_1(t)}{1 +12\eps_2}, g_1(t) \cdot  (1+12\eps_2)\right]$.  
 
Because $\bar{g}_1 \sim_{(\eps_1^4 , \eps_2)} g_1$, 
there exists $t' \in [ t /(1+\eps_1^4), t.(1+\eps_1^4)]$
and similarly  $\bar{g}_1 (t') \in  [g_1(t)/(1+\eps_2), g_1(t).(1+\eps_2)]$. There are four possibilities whether $t'>t$ or $t'\leq t$ and $\bar{g}_1 (t') > g_1(t)$ or $\bar{g}_1 (t') \leq g_1(t)$. Assume $t'>t$ and $\bar{g}_1 (t')> g_1(t)$, as the other cases are similar.

Notice that $t'.(1+3\eps_1)\leq t.(1+\eps_1^4).(1+3\eps_1)\leq t.(1+4\eps_1)$ and  similarly
$t/(1+4\eps_1) \leq t.(1+3\eps_1)\leq t'.(1+3\eps_1)$. Therefore if $u \notin  [ t /(1+4\eps_1), t.(1+4\eps_1)]$ then $u \notin  [ t' /(1+3\eps_1), t'.(1+3\eps_1)]$ as the second interval is included in the first.
Similarly 
$\bar{g}_1 (t') . (1+10\eps_2)\leq g_1(t).(1+\eps_2). (1+10\eps_2)\leq  g_1(t).(1+12\eps_2)$ and
$g_1 (t) /(1+12\eps_2)\leq g_1(t)/((1+\eps_2). (1+10\eps_2))\leq  \bar{g}_1(t')/(1+10\eps_2)$
Therefore  if $ g_2(u) \notin [g_1(t)/(1 +12\eps_2), g_1(t) \cdot  (1+12\eps_2)]$ then
 $ g_2(u) \notin [\bar{g}_1(t')/(1 +10\eps_2), \bar{g}_1(t') \cdot  (1+10\eps_2)]$   as the second interval is included in the first.
We conclude that $\bar{g}_1 \not\sim_{(3\eps_1, 10\eps_2)} g_2$. Hence the Tester $A_2$ rejects with high probability as it duplicates the Tester $A_1$.
\end{proof}

%If $g_1$ is   $(3\eps_1/2, \eps_2/4)$ half-stable  and  $(\gamma_1/(1+\eps_1),(1+\eps_2)^2\gamma_2 )$-decreasing

\subsubsection{Comparison between the frequency distributions of two marginals}
Assume each element of the stream is a tuple $(e_1,e_2,...e_d)$ where  each $e_i$ is an element of a domain of size $n$ and $d$ is even.  Consider two marginal frequency distributions $g_1$ as the frequency distributions of $(e_1,e_2,...e_{d/2})$, the projection on the first $d/2$ dimensions, and $g_2$ as the frequency distributions of $(e_{d/2 +1}...e_d)$, the projection on the last $d/2$ dimensions. Assume their support sizes are  similar.  We want to compare $g_1$ and $g_2$ as in the previous section. We may also  want to compare $g_1$ with a  reference distribution $f$.

As with the Tester $A_2$, we sample  both streams and for each substream $s_i^1$ and $s_i^2$ we have two tables $T_i^1$  and $T_i^2$ sorted by counter values. The first  table $T_i^1$ applies to the elements
 $(e_1,e_2,...e_{d/2})$ for each data item and let $c_{i,r}^1$ be the counter value in position $r$ of the table $T_i^1$. Similarly, the table $T_i^2$ applies to the elements
 $(e_{d/2 +1}...e_d)$ for each data item and let
 $c_{i,r}^2$  be the counter value in position $r$. The values of $K_1$ and $K_2$ are asymmetrical as in Theorem \ref{main3}.\\

%im_{(\eps_1, \eps_2)}g$ instead of $f=g$.\\

\begin{corollary}\label{st}

%\begin{T2}
\thmtextb
\end{corollary} 
\begin{proof}
The stream $s^1$ plays the role of  the reference function $f$ in Theorem \ref{main}. As in Theorem \ref{main3}, we introduce 
$\bar{c}_{i,r}^1=    Min _{j\leq i }\{ c_{j,r}^1 \}  $ and the Tester $A_2$ compares
$\bar{c}_{i,r}^1$ with $c_{i,r}^2 $ for the values $i-1, i, i+1$. Precisely, it tests if
$\bar{c}_{i,r}^1\not \simeq_{5.\eps_2}c_{i,r}^2 $.
The analysis of Theorem \ref{main3} is identical.
\end{proof}

\section{Conclusion}
We introduced a scale free distance between two frequency distributions, the relative version of the  Fr{\'e}chet distance.
We then studied how to verify a frequency function $g$ defined  by a stream of $N$ items among $n$ distinct items. We first proved a  $\Omega(n)$ lower bound on the space required in general. If we assume that  the frequency distribution $f$ and the frequency $g$ defined by the stream satisfy a
half-stability condition and  decrease fast enough, we presented a Tester that uses
 $O(\log^2 n\cdot \log\log n)$ space. We then generalized the Tester to compare the frequencies of two distinct streams.
 
 We will study the turnstile model~\cite{M05} with insertions and deletions  for  the bounded deletions model\footnote{In such a model,  the number $D$ of deletions  is related to the number $I$ of insertions:  $D\leq (1-1/\alpha)I$, for some constant $\alpha \geq 1$.} from \cite{J18} in some later work.
%\footnote{Our study of this model is deferred because the  bounded deletions model is studied in~\cite{Z22} but the algorithm therein has some issues currently in the process of being corrected.}. 
The approach generalizes to
 the {\em sliding window} model with insertions and deletions outside a window.
Notice that the sliding window model  is not a bounded deletion model, as $I/D$ tends to $ 1$ when $I $ goes to $\infty$.

%\appendix
\bibstyle{plain}
\bibliography{bibliographydegreedistribution}

\appendix

\section{Relative Fr\'echet distance}\label{a1}

 We introduced the notion of \emph{relative Fr\'echet distance} between two functions with definition \ref{definition:Frechet}. 
 The relative Fr{\'e}chet distance differs from the absolute Fr{\'e}chet distance. For example, consider two  families of step functions:

\begin{equation}
    f(i) =
    \begin{cases*}
     4b & if $i\leq 10c$ \\
      b      & if $i> 10c$ 
    \end{cases*}
     ~~~~g(i) =
    \begin{cases*}
     5b & if $i\leq 11c$ \\
      b      & if $i> 11c$ 
    \end{cases*}
  \end{equation} 
The absolute Frechet distance between $f$ and $g$ is $\max(b,c)$ (based on the $L_{\infty}$ metric) which is arbitrary large, whereas the relative Frechet distance is $(\eps_1,\eps_2)=(0.1,0.25)$, independent of $(b,c)$.
 
There is a different approach which generalizes the classical
(absolute) Fr{\'e}chet distance, based on the notion of {\emph coupling}, defined in \cite{EM94} and which we now recall. Here we also define the \emph{relative length} of a coupling.

\begin{definition}\label{def:coupling}
Let $f$ and $g$ be two functions with domain $\{ 1,\cdots ,n\}$. For $1\leq t\leq n$, consider the points  $u_t=(t,f(t))$ and $v_t=(t,g(t))$. 
A {\em coupling} between $f$ and $g$ is a sequence $(u_{a_1},v_{b_1}),(u_{a_2},v_{b_2}),\cdots ,(u_{a_m},v_{b_m})$ such that $a_1=1,b_1=1, a_m=n, b_m=n$, and for all $i$ we have $a_{i+1}\in \{ a_i, a_i+1\}$ and $b_{i+1}\in \{ b_i,b_i+1\}$. The coupling has \emph{relative length}  $(\eps_1,\eps_2)$ if for all $i$ we have $u_{a_i}\simeq_{(\eps_1,\eps_2)}v_{b_i}$. 
\end{definition}

Informally, with the coupling, to each point $u_t$ of the graph of $f$, we associate at least one point $v_{t'}$ of the graph of $g$, and vice-versa, in a monotone way: for all $t$ the points associated to $(t,f(t))$ are all to the left of the points associated to $(t+1,f(t+1))$, and the points associated to $(t,g(t))$ are all to the left of the points associated to $(t+1,g(t+1))$.

We now define the relative Fr{\'e}chet distance based on a coupling. 

\begin{definition}\label{Frechet}
Let $f$ and $g$ be two functions with domain $\{ 1,\cdots ,n\}$.  We say that $f$ and $g$ are $(\eps_1,\eps_2)$-close by coupling, denoted $f \sim_{(\eps_1,\eps_2)}^c g$,  if there exists a coupling of relative length at most $\eps_1,\eps_2$.
\end{definition}

If $f \sim_{(\eps_1,\eps_2)}^c g$ then $f \sim_{(\eps_1,\eps_2)} g$ based on the definition \ref{definition:Frechet}. Any point $(x,f(x)$ is $(\eps_1,\eps_2)$-close to a point
$(x',g(x'))$. We now show the converse and therefore establish that the definitions
\ref{definition:Frechet} and  
\ref{Frechet} are equivalent.

\begin{lemma}[Distance lemma]\label{sl}
If $f$  and $g$ are two functions describing  frequencies and such that 
for every point of $f$, there is a point of $g$ which is 
$(\eps_1,\eps_2)$-close and conversely. Then
$$f  \sim_{(\eps_1,\eps_2)} g$$
\end{lemma}

\begin{proof}
Given a point $u_i=(i,f(i))$ of $f$, we first claim that the set $S_i$ of $j$ such that $v_j=(j,g(j))$ satisfies $v_j\simeq_{(\eps_1,\eps_2)} u_i$ is an non-empty interval, as shown in figure \ref{fig:prooflemma5}. Indeed, it is non-empty by assumption. Let $j_{\min}$ and $j_{\max}$ be its minimum and maximum elements respectively. Then for every $j\in [j_{\min}, j_{\max}]$, we have $i/(1+\eps_1)\leq j_{\min}\leq j$ and $j\leq j_{\max}\leq i(1+\eps_1)$, so $j\simeq_{\eps_1} i$; and by monotonicity, $f(j)\leq f(j_{\min})\leq (1+\eps_2)g(i)$ and $g(i)/(1+\eps_2) \leq g(j_{\max})\leq g(j)$, so $f(i)\simeq_{\eps_2}g(j)$, proving the claim. 

Let $S_i=[\ell_i,r_i]$.  We also claim that the sequence $(\ell_i)_i$ and $(r_i)_i$ are monotone non-decreasing. Indeed, assume, for a contradiction, that $\ell_i>\ell_{i+1}$. Then $i<i+1<\ell_{i+1}(1+\eps_1)$ and $i>\ell_i/(1+\eps_1)> \ell_{i+1}/(1+\eps_1)$, so $i\simeq_{\eps_1} \ell_{i+1}$; moreover, $g(\ell_{i+1}) \leq f(i+1)(1+\eps_2) \leq f(i)(1+\eps_2)$, and
 $g(\ell_{i+1})(1+\eps_2) \geq g(\ell_i) (1+\eps_2) \geq  f(i) $, so $u_i\simeq_{(\eps_1,\eps_2)}v_{\ell_{i+1}}$, a contradiction. The proof of the monotonicity of $(r_i)$ is similar. 
 
 Moreover, the collection of intervals $(S_i)_i$ covers $[1,n]$ because every point of $g$ is $(\eps_1,\eps_2)$-close to some point of $f$. 

The coupling then simply consists of the pairs $$\{ (i,j): \max(\ell_i,r_{i-1}) \leq j\leq r_i\},$$ in lexicographic order, i.e. the red edges of figure \ref{fig:prooflemma5}. Let us verify that this is a correct coupling sequence. Since $r_{i-1}\leq r_i$, every $i$ belongs to at least one pair. Every $j$ will appear in the pair $(i,j)$ where $i$ is minimum such that $r_i\geq j$. Such an $i$ exists because every $j$ belongs to at least one $S_i$. From one element of the sequence to the next, we either keep $i$ unchanged and move from one element of $S_i$ to the next element of $S_i$, incrementing the count by 1 on $g$; or we switch from $S_i$ to $S_{i+1}$, incrementing $i$ and possibly incrementing $j$ by one as well, in the case in which $r_i\notin S_{i+1}$. Thus this forms a correct coupling sequence such that $f  \sim_{(\eps_1,\eps_2)} g$.

\begin{figure}
\includegraphics[width=12cm, height=7cm]{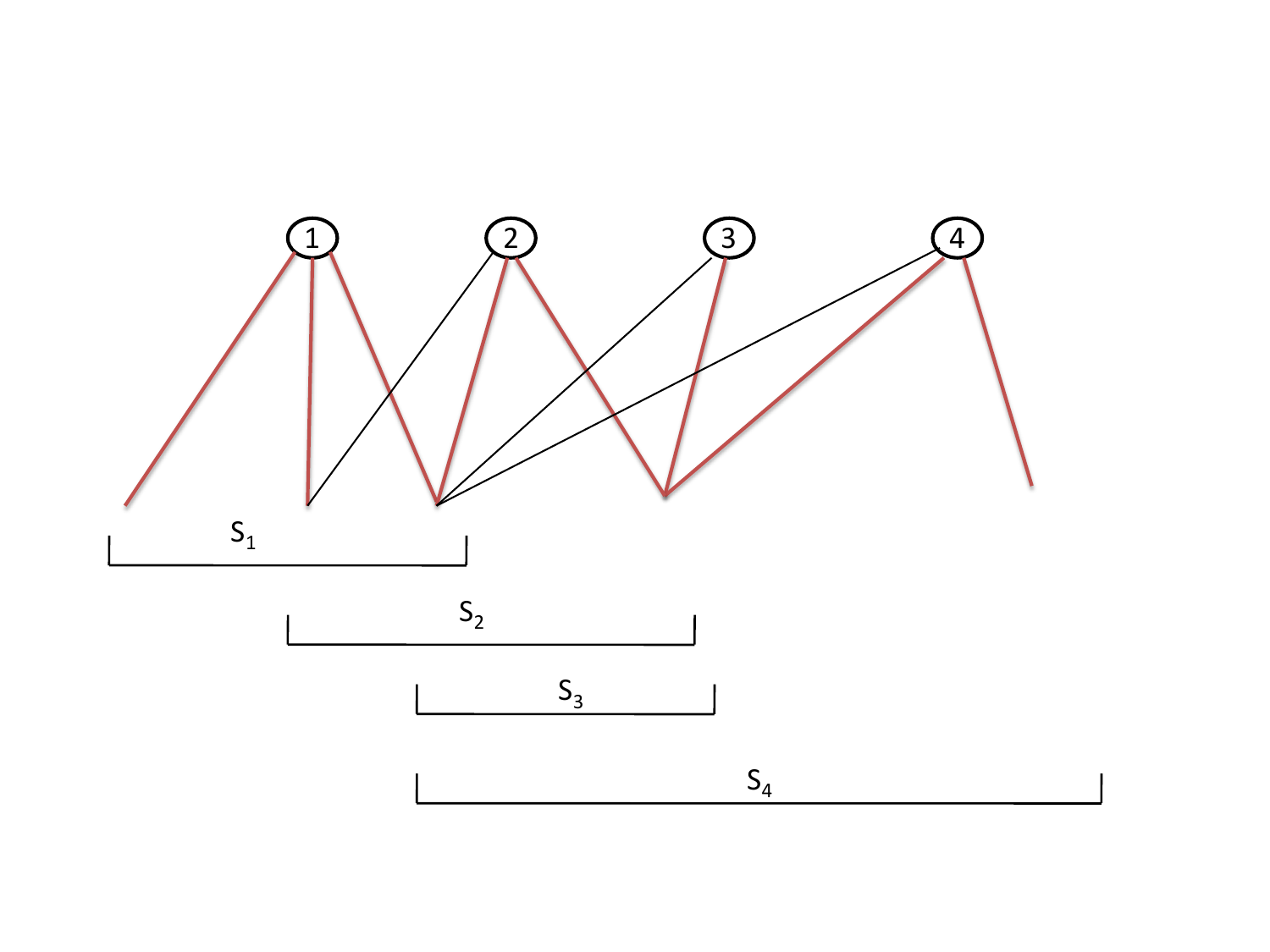}
\caption{Proof of lemma \ref{sl}. The thick red edges are the coupling edges.}
\label{fig:prooflemma5}
\end{figure}
\end{proof}

\section{Property of Half stable functions}\label{hsf}

 The following Lemma provides an equivalent characterization of half-stable frequency functions. It states that we can partition the domain $\{ 1,2,\ldots ,n\}$ into {\em Intervals} where the relative variations of the function $f$ are bounded.
 
\begin{lemma}\label{lemma:stepcompatible}

If $f$ is an   {\em $ (\eps_1,\eps_2)$-half-stable frequency function,} then there exists a partition of $\{ 1,2,\ldots ,n\}$ into {\em Intervals} $[\ell_j,r_j]$ with $j\geq 1$ such that for all $j$:
\begin{itemize}
\item (a) ~~$\ell_{j+1}=r_j +1 $ and $r_j \geq  \lfloor  (1+\eps_1)\ell_j \rfloor$ ; and
\item (b) ~~ $f(r_j) \geq f(\ell_j)/(1+3\eps_2)$.
\end{itemize}
If there exists a partition of $\{ 1,2,\ldots ,n\}$ into {\em Intervals} $[\ell_j,r_j]$ satisfying (a) and (b), then $f$ is $(\eps_1,4\eps_2)$-half-stable.
\end{lemma}
\begin{proof}
The intervals are defined in a 2-step process. The first step is greedy: let $(x_i)_{i\geq 1}$ denote the sequence of distinct values of $\lceil (1+\eps_1/3)^{j} \rceil$ and $y_i=x_{i+1}-1$ (or $y_i=n$ if $i$ is the last term of the sequence). Using the fact that $f$ is $(\eps_1,\eps_2)$-half-stable, let $R_i$ denote the $(\eps_1,\eps_2)$ rectangle containing $(x_i,f(x_i))$ and note that for $i$ large enough $R_i$ must contain $(x_{i+1},f(x_{i+1}))$ or $(x_{i-1},f(x_{i-1}))$ (otherwise its relative horizontal span would be less than $(1+\eps_1/3)^2<1+\eps_1$, so it intersects $R_{i-1}$ or $R_{i+1}$. Extract greedily a maximal subsequence $R_{i_1},R_{i_2},R_{i_3},\cdots$ of $R_i$'s containing $R_1$ and among which no two rectangles intersect. The sequence $\ell_j$ then consists of the left endpoints of the rectangles in that subsequence. Finally, we set $r_j=\ell_{j+1}-1$ (except that we set $r_j=n$ for the last interval). 

Each interval $[\ell_j,r_j]$ contains at least the horizontal span of a rectangle $R_{i_j}$ of the subsequence, so the first property holds: $\ell_{j+1}>(1+\eps_1)\ell_j$.
Consider the rightmost rectangle $R_k$ that intersects $R_{i_j}$.  All the points $(t,f(t))$ with $\ell_j\leq t\leq r_j$ are in the horizontal span of $R_{k'}\cup R_{i_j}\cup R_k$. The vertical span is therefore at most that of two $(\eps_1,\eps_2)$ rectangles, i.e. $$f(\ell_j)\leq (1+\eps_2)^2 f(r_j) < (1+4\eps_2)f(r_j)$$ 

The other direction is straightforward.
\end{proof}
%\subsection{The SpaceSaving algorithm with insertions only~\cite{MAE2005} }\label{ssaving}

%%%%%%%%%%%%%%%%%%%%%%%%%

%\subsection{Exact and approximate solutions}

%Appendix \ref{amo} generalizes to other streaming models.

\end{document}